\documentclass[sigconf]{acmart}

\usepackage{multirow}
\usepackage{multicol}
\usepackage{siunitx}
\usepackage{float}
\usepackage{subcaption}
\usepackage{xcolor}
\usepackage{enumitem}
\usepackage{titlecaps}

\usepackage{lipsum}
\usepackage{nicefrac}
\usepackage[ruled,vlined,linesnumbered]{algorithm2e}
\usepackage{soul}

\newtheorem{property}{Property}[section]
\theoremstyle{plain}
\newtheorem{theorem}{Theorem}[section]
\theoremstyle{remark}

\usepackage{pifont}

\newcommand{\htree}{\hq}
\newcommand{\htreename}{tree hierarchy}
\newcommand{\otree}{\preceq_H}
\newcommand{\ch}{\hu}
\newcommand{\dch}{\text{DH}_U}
\newcommand{\chname}{update hierarchy}
\newcommand{\lab}{L}

\newcommand{\sgd}{d_{\ch}^{[\!w,v]}}

\newcommand{\hq}{H_Q}
\newcommand{\subtree}{T_\downarrow}
\newcommand{\hu}{H_U}
\newcommand{\om}{DHL}

\newcommand{\hh}{H2H-Index}
\newcommand{\ou}{\preceq_H}

\newcommand{\anc}{anc}
\newcommand{\desc}{desc}
\newcommand{\nup}{N^+}
\newcommand{\ndown}{N^-}

\newcommand*{\lidx}[1]{\tau({#1})}

\AtBeginDocument{%
  }

\setcopyright{rightsretained}
\acmJournal{PACMMOD}
\acmYear{2025} \acmVolume{3} \acmNumber{1 (SIGMOD)} \acmArticle{35} \acmMonth{2} \acmPrice{}\acmDOI{10.1145/3709685}

\makeatletter
\gdef\@copyrightpermission{
   \href{https://creativecommons.org/licenses/by/4.0/}{\includegraphics[width=0.10\textwidth]{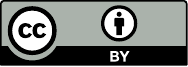}}\\
   \href{https://creativecommons.org/licenses/by/4.0/}{This work is licensed under a Creative Commons Attribution International 4.0 License.}
  \vspace{1pt}
}
\makeatother

\begin{document}

\title{Dual-Hierarchy Labelling: Scaling Up Distance Queries on Dynamic Road Networks}


\author{Muhammad Farhan}
\affiliation{%
  \institution{Australian National University}
  \city{Canberra}
  \country{Australia}}
\email{muhammad.farhan@anu.edu.au}

\author{Henning Koehler}
\affiliation{%
  \institution{Massey University}
  \city{Palmerston North}
  \country{New Zealand}}
\email{h.koehler@massey.ac.nz}

\author{Qing Wang}
\affiliation{%
  \institution{Australian National University}
  \city{Canberra}
  \country{Australia}}
\email{qing.wang@anu.edu.au}

\begin{abstract}
Computing the shortest-path distance between any two given vertices in road networks is an important problem. A tremendous amount of research has been conducted to address this problem, most of which are limited to static road networks. Since road networks undergo various real-time traffic conditions, there is a pressing need to address this problem for dynamic road networks. Existing state-of-the-art methods incrementally maintain an indexing structure to reflect dynamic changes on road networks. However, these methods suffer from either slow query response time or poor maintenance performance, particularly when road networks are large. In this work, we propose an efficient solution \emph{Dual-Hierarchy Labelling (DHL)} for distance querying on dynamic road networks from a novel perspective, which incorporates two hierarchies with different but complementary data structures to support efficient query and update processing. Specifically, our proposed solution is comprised of three main components: \emph{query hierarchy}, \emph{update hierarchy}, and \emph{hierarchical labelling}, where \emph{query hierarchy} enables efficient query answering by exploring only a small subset of vertices in the labels of two query vertices and \emph{update hierarchy} supports efficient maintenance of distance labelling under edge weight increase or decrease. We further develop dynamic algorithms to reflect dynamic changes by efficiently maintaining the update hierarchy and hierarchical labelling. We also propose a parallel variant of our dynamic algorithms by exploiting labelling structure which aligns well with parallel processing. We evaluate our methods on 10 large road networks and it shows that our methods significantly outperform the state-of-the-art methods, i.e., achieving considerably faster construction and update time, while being consistently 2-4 times faster in terms of query processing and consuming only 10\%-20\% labelling space.
\end{abstract}

\maketitle

\section{Introduction}
Typically, a road network is modeled as a dynamic weighted graph $G=(V, E, \omega)$, where vertices $V$ represent intersections, edges $E$ represent roads between intersections, and $\omega$ assigns weights to edges, such as travel time. In a dynamic road network, it is commonly assumed that vertices and edges remain intact, while the weights of edges may continuously change due to real-time traffic conditions, such as traffic congestion, accidents, or road closures. Given any arbitrary pair $(s,t)$ of vertices, the \emph{distance querying} problem is to answer the up-to-date shortest-path distance between $s$ and $t$ under these dynamic changes. 

As a fundamental building block, distance querying plays a crucial role in road applications, such as GPS navigation \cite{goldberg2005computing}, route planning \cite{fan2010improvement}, traffic monitoring \cite{kriegel2007proximity}, and POI recommendation \cite{yawalkar2019route}. These applications often demand low-latency responses due to the need of computing thousands to millions of distance queries per second, as part of more complex tasks, which itself needs to be solved frequently. Examples include matching taxi drivers with passengers, optimizing delivery routes with multiple pick up and drop off points that can change dynamically, or providing recommendation on k-nearest POIs to customers. For instance, real-time navigation systems like Google Maps and Waze use preprocessing techniques such as contraction hierarchies~\cite{geisberger2008contraction} to accelerate the computation of shortest-path distance queries, enabling optimal route suggestions based on current traffic conditions. These conditions are obtained through various sources, such as crowdsourcing (where data like speed and location is collected from users who have the app open on their devices), road sensors, local transportation authorities, and real-time reports from users. This multifaceted approach ensures that traffic data is updated as frequently as possible, often multiple times per minute \cite{zheng2010understanding,ouyang2020efficient}. In a similar vein, companies offering ride-hailing services like Uber and Lyft rely on millions of real-time distance queries to efficiently match drivers with passengers and minimize wait times \cite{zhang2021efficient,huang2021learning}. To meet user demands, these systems commonly use a hybrid client-server model. Complex real-time computations and traffic updates are handled by the server, while simpler tasks are managed on the client side. This allows for efficient and responsive navigation, even when the internet connection is weak or unavailable.

\begin{figure*}
\centering
\includegraphics[width=\textwidth]{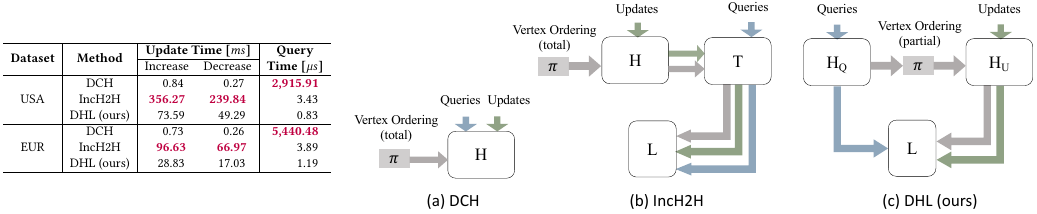}
\vspace{-0.2cm}
\caption{A high-level comparison of our proposed method DHL with state-of-the-art methods: (a) DCH~\cite{ouyang2020efficient} and (b) IncH2H~\cite{zhang2022relative}. DCH processes both queries and updates based on CH (i.e., $H:= CH$), IncH2H extends DCH to construct another tree hierarchy $T$ to speed up queries, and DHL uses a tree hierarchy as $H_Q$ and CH as $H_U$ to process queries and updates, separately; $L$ refers to a distance labelling. The grey, green and blue arrows indicate the processes of construction, updating, and querying, respectively.}
\label{fig:our_frameworks}\vspace{-0.2cm}
\end{figure*}

\vspace{0.1cm}
\noindent\textbf{Present work.~}
In this paper, we study the distance querying problem in dynamic road networks, aiming to maintain accurate distance information efficiently while still achieving superior query performance.  
\emph{This task is challenging, as query time and update time are typically considered a trade-off, where improving one often comes at the expense of the other.} 

Previously, several methods~\cite{geisberger2012exact,ouyang2020efficient,wei2020architecture,zhang2021dynamic,zhang2022relative,chen2021p2h} have been developed to incrementally maintain precomputed information, ensuring fast and correct responses to distance queries. These methods generally fall into two main categories:~\begin{itemize}
    \item[(1)] \emph{Shortcut-based methods}~\cite{geisberger2012exact,ouyang2020efficient,wei2020architecture}, which focus on maintaining precomputed auxiliary edges, referred to as \emph{shortcuts}, between vertices by applying vertex contraction in terms of a given vertex ordering, 
e.g., contraction hierarchy (CH)~\cite{geisberger2008contraction}. Such shortcuts help reduce query search space. 
\item[(2)] \emph{Hub-based labelling methods}~\cite{zhang2021dynamic,zhang2022relative,chen2021p2h}, which maintain distance labelling in addition to shortcuts so that distance queries can be answered by searching distance labels directly. However, these methods require constructing an additional tree decomposition based on CH. 
\end{itemize} In a nutshell,
shortcut-based methods require less space and shorter update time, but query time can be orders of magnitude longer than hub-based labelling methods. In contrast, hub-based labelling methods answer distance queries significantly faster at the cost of requiring
additional space for storing distance labels and longer update time
due to the combined maintenance of distance labels and shortcuts. The table in Figure~\ref{fig:our_frameworks} presents the empirical results on two large real-world road networks: one from USA~\cite{demetrescu2009shortest} and the other from Europe~\cite{ptvplanung}.

To tackle the challenge of efficiently maintaining accurate distance information without compromising query performance, we make two crucial observations: 

\begin{itemize}
    \item \textbf{Separation of concerns}: Queries and updates are distinct types of operations. Therefore, different data structures can be employed to efficiently support distance queries and updates separately.
    \item \textbf{Vertex ordering}: Finding a vertex ordering based on the centrality of vertices in a road network plays a vital role in accelerating the performance of both queries and updates.
\end{itemize}
Previous works either use the same data structure for both queries and updates (e.g., DCH~\cite{ouyang2020efficient}) or use two data structures where the one for queries depends on the other to reflect updates (e.g., IncH2H~\cite{zhang2022relative}). Further, these methods, including DCH and IncH2H, all utilize a total vertex ordering based on general heuristics, such as minimum degree~\cite{berry2003minimum}, without considering the specific properties of a road network. Figure~\ref{fig:our_frameworks}(a)-(b) presents a high-level description of the methods DCH and IncH2H.

These observations guide us to develop a novel framework, {\emph{Dual-Hierarchy Labelling} (DHL)}, which aims to balance the trade-off between query time and update time while maintaining efficient and accurate distance querying in dynamic road networks. As depicted in Figure~\ref{fig:our_frameworks}(c), the {DHL} framework has three cornerstones: $(\langle \hq, \hu \rangle, L)$, where $\langle \hq, \hu \rangle$ are two hierarchies integrated through {a vertex partial ordering $\pi$}, and $L$ is a hierarchical labelling satisfying the 2-hop cover property~\cite{cohen2003reachability}.  
More specifically, $\hq$ is a static \emph{query hierarchy} designed to accelerate the efficiency of distance queries, while $\hu$ is a dynamic \emph{update hierarchy} aimed at enhancing the efficiency of distance maintenance in dynamic road networks.
Both $\hq$ and $\hu$ respect the partial order $\pi$ on vertices. Unlike previous work, the vertex ordering $\pi$ in our method is obtained by ordering vertices in terms of their occurrences in the minimum cuts of recursive partitions of a road network, reflecting their centrality in the entire network. In doing so, $\hq$ supports \emph{efficient querying} by reducing the search space of labels when using $L$, while $\hu$ supports \emph{efficient maintenance} under dynamic changes, ensuring that $L$ dynamically maintains correct distance information. Hence, distance queries are answered via $\hq$ and $L$, while dynamic changes are maintained via $\hu$ and $L$.

Based on the {DHL} framework, we further explore the following:

\begin{itemize}[leftmargin=*]
    \item \emph{\underline{ Design Choices:}~}~We instantiate the DHL framework by employing a tree hierarchy for $\hq$ and a variant of the contraction hierarchy for $\hu$. Our design addresses the inherent limitations of existing shortcut-based and hub-based methods, such as DCH and IncH2H, by extracting a vertex partial ordering from $\hq$ to construct $\hu$. $L$ is designed to consist of a \emph{distance schema} $\Gamma$ determined by $\hq$ and a \emph{distance map} $\gamma$ determined by $\hu$. The distance schema $\Gamma$ captures the tree-like structural information of labelling in terms of $\hq$, while the distance map $\gamma$ maintains distance information efficiently by propagating dynamic changes through $\hu$. Furthermore, our hierarchical labelling $L$ preserves distances within subgraphs rather than the entire graph, while still exhibiting the 2-hop cover property for query efficiency. Restricting distances to subgraphs has the distinct advantage that a weight update only affects distance entries in the subgraph where the updated edge lies, leading to faster maintenance.\vspace{0.1cm}

    \item \emph{\underline{Dynamic Algorithms:}~} We propose dynamic algorithms that can efficiently maintain the update hierarchy $\hu$ and the hierarchical labelling $L$ to reflect edge weight decrease and increase on road networks. Specifically, under dynamic changes, we first maintain $\hu$ to find affected shortcuts and then use these shortcuts as a starting point to update distance entries via the distance map $\gamma$. By exploring the structure of hierarchical labelling $L$, we also propose parallel variants of our dynamic algorithms which further improve maintenance performance by means of parallelizing searches w.r.t. multiple ancestors.\vspace{0.1cm}

    \item \emph{\underline{Theoretical Analysis}:~} We conduct theoretical analysis to prove the correctness of our algorithms through establishing a connection between shortcuts in $\hu$ and distance entries in $L$ based on $\hq$. We also theoretically prove that restricting distances stored in the hierarchical labelling $L$ to subgraphs, rather than distances in the entire graph, does not affect the 2-hop cover property of $L$. Furthermore, we analyse the complexity bounds of our dynamic algorithms for both edge increase and edge decrease. 
\end{itemize}
We have conducted extensive experiments to evaluate the performance of our algorithms on 10 real-world large road networks, including the whole road network of USA and western Europe road network. The results show that our algorithms consistently outperforms the state-of-the-art method IncH2H on all of these road networks in both query time and update time. Figure~\ref{fig:our_frameworks} presents some results on two largest datasets (refer to Section~\ref{performance} for details). In general, our method is about 3-4 times faster than IncH2H in terms of update time, while being 2-4 times faster in terms of query processing and consuming only 10\%-20\% labelling space. 
\section{Related Work}
The classical approach for computing the shortest-path distance between a pair of vertices is to use Dijkstra's algorithm \cite{tarjan1983data}, which has a time complexity of $O(|E|+|V|log|V|)$ and may traverse an entire network to answer queries for vertices that are far apart from each other. Specifically, given a source vertex $s$ and a target vertex $t$, Dijkstra's algorithm~\cite{tarjan1983data} traverses vertices in ascending order of their distances from $s$ until $t$ is reached. Bidirectional Dijkstra’s algorithm~\cite{wu2012shortest}, a variation of Dijkstra’s algorithm, performs two Dijkstra's searches simultaneously from $s$ and $t$, traversing vertices in ascending order of their distances to $s$ and $t$, respectively. A* algorithm \cite{hart1968formal} estimates for each visited vertex $v$ the heuristic distance value from $v$ to the target vertex $t$ and use it as the searching guidance. A* algorithm reduces Dijkstra’s search space to a smaller conceptual ellipse \cite{bast2016route}. Since search-based methods do not require any precomputed information and have no construction or maintenance cost, they are naturally adopted to dynamic road networks. However, these methods are very inefficient in query processing and vulnerable to large search space, particularly on large road networks (e.g., millions of vertices). 

Below, we focus on discussing two lines of work that leverage and maintain precomputed information for answering distance queries on dynamic road networks. 

\vspace{0.2cm}
\noindent\emph{\underline{Shortcut-Based Methods}.~}~
To improve query time, which is crucial for real-world applications, several works \cite{ouyang2020efficient,geisberger2012exact,geisberger2008contraction,zhu2013shortest,wei2020architecture} precompute shortcuts between vertices. \emph{Contraction hierarchy} (CH)~\cite{geisberger2008contraction} is the most widely used method in this direction. Given a pre-defined vertex order $\tau$, CH performs vertex contraction on each vertex $v$ one by one following $\tau$ to add shortcuts among its neighbors which preserve their distance information through $v$. Later, 
CH was extended to the dynamic setting  \cite{ouyang2020efficient,wei2020architecture,geisberger2012exact}, which can be divided into two categories: (1) \emph{vertex-centric methods}; (2) \emph{shortcut-centric methods}. The vertex-centric methods \cite{geisberger2012exact} first identify affected vertices and then re-contract them using the original vertex order to update shortcuts. The shortcut-centric methods \cite{ouyang2020efficient,wei2020architecture} first identify affected shortcuts and then update the weights of these shortcuts. The former requires the minimality of shortcuts during recontraction under the \emph{shortest distance constraint}~\cite{geisberger2008contraction}, which is highly inefficient because new or old shortcuts may need to be added or removed accordingly. The latter allows redundant shortcuts to avoid insertion or deletion of shortcuts during maintenance and updates only the weights of affected shortcuts due to the \emph{minimum weight property}~\cite{ouyang2020efficient}. However, the shortcut-centric methods may add much more shortcuts than the vertex-centric methods, making it difficult to scale to graphs with large treewidth due to potentially extremely dense structures.

\vspace{0.2cm}
\noindent\emph{\underline{Hub-Based Labelling Methods}.~} Recent works \cite{zhang2021dynamic,zhang2022relative,chen2021p2h} are mostly hub-based labelling methods, which precompute \emph{distance labels} that capture distance information between all pairs of vertices in $G$. Then, only distance labels are searched at query time to compute distances. These works are typically built upon \emph{hierarchical 2-hop index} (H2H-Index) \cite{ouyang2018hierarchy} which can reduce search to a subset of labels $L(s)$ and $L(t)$ for two querying vertices $s$ and $t$ by exploiting a vertex hierarchy using tree decomposition. 

Specifically, they extend H2H-Index to dynamic road networks by incrementally maintaining H2H-Index to reflect edge weight updates on $G$.  DynH2H~\cite{chen2021p2h} is the first dynamic algorithm which maintains H2H-Index in order to efficiently answer distance queries under dynamic changes on road networks. Zhang et al.~\cite{zhang2021dynamic} propose \emph{dynamic tree decomposition based hub labelling} (DTDHL), an optimized version of DynH2H. Nonetheless, DTDHL may take seconds even for a single  change, failing to scale to large road networks. 
Recently, Zhang et al. \cite{zhang2022relative} studied the theoretical boundedness of dynamic CH index \cite{ouyang2020efficient,wei2020architecture} and dynamic H2H index \cite{chen2021p2h,zhang2021dynamic} for edge weight increase and decrease separately. IncH2H has shown to achieve the state-of-the-art performance on dynamic road networks. However, it suffers from maintaining a huge index which is constructed upon CH using a total ordering of vertices based on \emph{minimum degree heuristic} \cite{berry2003minimum}. 
\section{State-of-the-Art Solutions}
Let $G=(V, E, \omega)$ be a road network where $V$ is a set of vertices and $E\subseteq V\times V$ is a set of edges. Each edge $(u, v) \in E$ is associated with a non-negative weight $\omega(u, v) \in \mathbb{R}_{\geq 0}$. A simple path $p$ is a sequence of distinct vertices $(v_1, v_2, \dots, v_k)$ where $(v_{i}, v_{i+1}) \in E$ for each $i\in[1, k)$. The weight of a path $p$ is defined as $\omega(p) = \sum_{i=1}^{k - 1} \omega(v_i, v_{i+1})$. A shortest path $p$ between two arbitrary vertices $s$ and $t$ is a path starting at $s$ and ending at $t$ such that $\omega(p)$ is minimised. The \emph{distance} between $s$ and $t$, denoted as $d_G(s,t)$, is the weight of any shortest path between $s$ and $t$. 
We use $N(v)$ to denote the set of neighbors of a vertex $v \in V$, i.e. $N(v) = \{(u, \omega(u, v)) \mid u\in V, (u, v) \in E \}$, and 
$V(G)$ and $E(G)$ the set of vertices and edges in $G$, respectively.
We assume that $G$ is undirected (unless explicitly indicated otherwise) and treat dynamic changes on $G$ as edge weight updates~\cite{ouyang2018hierarchy,zhang2022relative}.

\begin{figure}
\centering
\includegraphics[width=0.4\textwidth]{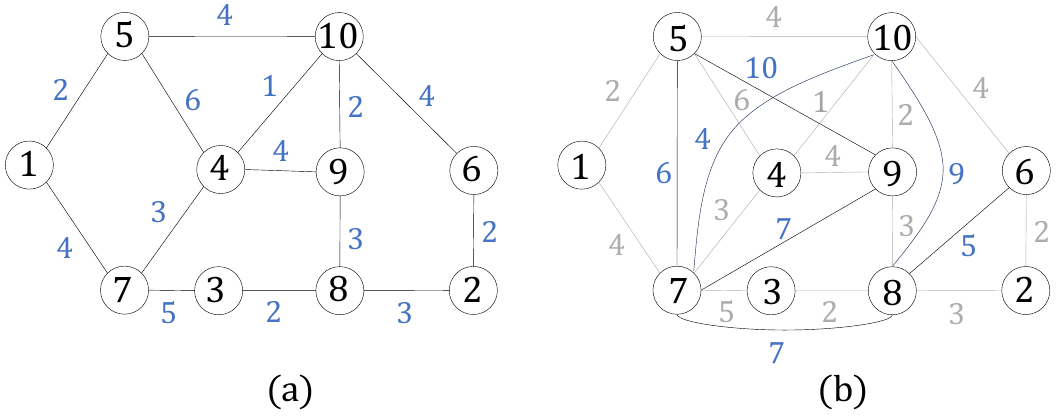}\vspace{-0.2cm}
\caption{(a) A road network $G$; (b) A shortcut graph $S_G$ of $G$.}
\label{fig:ch}\vspace{-0.2cm}
\end{figure}

Below, we discuss two state-of-the-art methods for dynamic road networks: (1) Dynamic Contraction Hierarchy (DCH)~\cite{ouyang2020efficient}, and (2) Incremental Hierarchical 2-Hop (IncH2H)~\cite{zhang2022relative}.

\subsection{Dynamic Contraction Hierarchy}\label{subsec:sota-dynamic}

\emph{Dynamic Contraction Hierarchy} (DCH) \cite{ouyang2020efficient} is built upon a variant \cite{geisberger2012exact} of the original Contraction Hierarchy (CH) \cite{geisberger2008contraction}. This variant ensures that the presence of shortcuts is weight-independent, albeit at the cost of adding more shortcuts compared to the original CH. This trade-off significantly improves update performance.

\vspace{0.1cm}
\noindent\emph{\underline{Indexing}.~} Given a graph $G$ and a total ordering $\pi$ over $V(G)$, DCH constructs a shortcut graph $S_G$ of $G$ by adding shortcuts as follows. A shortcut $(v_1, v_k)$ is added between two vertices $v_1$ and $v_k$ if and only if $G$ contains a valley path \cite{wei2020architecture}. A \emph{valley path} is defined to be a simple path $p = (v_1, v_2, \dots, v_k)$ satisfying $\pi(v_i) < \min\{\pi(v_1), \pi(v_k)\}$ for $\forall 1 < i < k$. In other words, the rank of any intermediate vertex of $p$ is lower than the rank of its endpoints. The weight of the shortcut $(v_1, v_k)$ is defined as the weight of the valley path $p$, i.e., $\omega(v_1,v_k)=\omega(p)$.

For each vertex $v$, let $N^+(v)$ be its upward neighbors in $S_G$ which are ranked higher than $v$, i.e., $N^+(v) = \{u \mid (u,v)\in E(S_G) \wedge \pi(u) > \pi(v)\}$, and $N^-(v)$ be its downward neighbors in $S_G$ which are ranked lower than $v$, i.e., $N^-(v) = \{u \mid (u,v)\in E(S_G) \wedge \pi(u) < \pi(v)\}$. Shortcuts in DCH satisfy the \emph{minimum weight property}~ \cite{ouyang2020efficient}.

\begin{property}\label{lab:weight_property} For any $(u,v)\in E(S_G)$, the following  holds:
\begin{align}\label{eq:weight_property}
\omega(u,v)=\min\{{\omega_G(u,v)}, \omega(x, u) + \omega(x, v) \nonumber\mid \\x\in N^-(u)\cap N^-(v) \}.
\end{align} 
\end{property}

Here $\omega_G(u,v)$ denotes the weight of edge $(u,v)$ in $G$ if it exists, or $\infty$ otherwise.
Note that every edge in $G$ is a valley path, and thus a shortcut in $S_G$, and that its weight in $G$ and $S_G$ may differ.

\vspace{0.1cm}
\noindent\emph{\underline{Updates and Queries}.~} DCH maintains shortcuts affected by edge weight decrease and increase separately. Let $(u,v)\in E(G)$ be an updated edge with $\pi(u)<\pi(v)$. For each $x\in N^+(u)$, if the weight of $(u,v)$ is decreased and $\omega( v,x)>\omega(u,v)+\omega(u,x)$, then the weight of shortcut $\omega(v,x)$ is updated to $\omega(u,v)+\omega(u,x)$; if the weight of $(u,v)$ is increased and $\omega(v,x)=\omega(u,v)+\omega(v,x)$, then the weight of shortcut $\omega(v,x)$ is updated based on Equation \ref{eq:weight_property}.
When answering a distance query between two vertices $s,t\in V(G)$, a bidirectional Dijkstra's search from $s$ to $t$ is performed over $S_G$ with the restriction that edges are only searched in an upward direction w.r.t. the vertex order $\pi$.

\begin{example}
Figure~\ref{fig:ch}(b) illustrates a shortcut graph $S_G$ obtained from an example road network $G$ shown in Figure~\ref{fig:ch}(a) following a total vertex ordering
$1{<}2{<}3{<}4{<}5{<}6{<}7{<}8{<}9{<}10$. Consider a query pair $(6, 9)$ on the road network $G$ in Figure~\ref{fig:ch}. A bidrectional search from $6$ and $9$ is conducted such that only the edges $(6,10),(6,8),(8,9)$ are considered in the search from vertex $6$, and only the edge $(9, 10)$ is considered in the search from vertex $9$ because they lead to vertices with higher ranks.
Thus we get $d_G(6,9) = \min\{ \omega(6,10)+\omega(9,10), \omega(6,8) + \omega(8,9) \} = 6$.
\end{example}

\subsection{Incremental Hierarchical 2-Hop}
\emph{Incremental Hierarchical 2-Hop} (IncH2H)~\cite{zhang2022relative} is built upon \hh~~\cite{ouyang2018hierarchy} to incrementally maintain distance labels and shortcuts for distance queries on dynamic road networks.  

\vspace{0.1cm}
\noindent\emph{\underline{Indexing}.~}\hh~ is a 2-hop labelling constructed using tree decomposition based on DCH. Let $S_G$ be a shortcut graph over $G$~\cite{ouyang2020efficient}. \hh~ first constructs a tree decomposition $T$ based on $S_G$ and $\pi$ as follows. Each vertex $u\in V(S_G)$ corresponds to a tree node in $T$. For a vertex $u\in V(S_G)$ with a non-empty $N^+(u)$, the vertex $v\in N^+(u)$ with lowest rank $\pi(v)$ is assigned as the parent of the tree node of $u$. This construction ensures that for each $u\in V(S_G)$, all vertices in $N^+(u)$ correspond to tree nodes that are ancestors of $u$ in $T$.
Another property of $T$ is that every shortest path between any two vertices $s$ and $t$ in $G$ must pass through at least one vertex in $\{u\}\cup N^+(u)$, where the tree node of $u$ is the lowest common ancestor of $s$ and $t$ in $T$. 
Then, a 2-hop labelling is constructed using $T$. The label $L(v)$ of each vertex $v$ consists of three arrays: (i) \emph{ancestor array} $[w_1,\dots, w_k]$ representing the path from the root to $v$ in $T$, (ii) \emph{distance array} $[\delta_{vw_1},\dots, \delta_{vw_k}]$ where $\delta_{vw_i}=d_G(v,w_i)$ and $\{w_1,\dots, w_k\}$ is the set of vertices that are ancestors of $v$ in $T$, and (iii) \emph{position array} $[i_1,\dots, i_k]$ storing positions of $\{w_1,\dots, w_k\}$ in $T$ which are their depths in $T$. 

\vspace{0.1cm}
\noindent\emph{\underline{Updates and Queries}.~} \hh~ is dynamically maintained 
in two phases: 1) \emph{shortcut maintenance} 
identifies affected shortcuts in $S_G$ and updates their weights; 2) \emph{labelling maintenance} updates distance labels based on affected shortcuts in $S_G$.
For each affected shortcut $(u,w)\in E(S_G)$ with $\pi(u)<\pi(w)$, it finds all ancestors $a$ of $w$ for which the distance between $u$ and $a$ has changed, and updated the corresponding values in the distance array. Finally, descendants $x$ of $u$ are processed to identify pairs $(x,a)$ whose distance has changed, and updates distances in $L(x)$.

Given two vertices $s,t\in V(G)$, their lowest common ancestor $lca(s, t)$ in $T$ is first searched and then $d_G(s,t)$ is computed as
\begin{align}\label{lab:h2h_query}
d_G(s,t)=\min\{ & L(s).dist(i) + L(t).dist(i) \mid  \\\nonumber
& i\in L(x).pos, x= lca(s, t)\}.
\end{align}

\begin{example}
Figure~\ref{fig:h2h} illustrates a tree decomposition $T$ and the \hh~ $L$ for a road network shown in Figure~\ref{fig:ch}(a). $L(1)$ stores an ancestor array $[10,9,8,7,5,1]$ containing all ancestors of $1$ in $T$, a distance array $[6,8,11,4,2,0]$ storing the distances from vertex $1$ to its ancestors, and a position array $[4,5,6]$ storing the positions in the ancestor array of $1$ for $\{7,5,1\}$, the vertices inside the tree node of $1$ in $T$. Suppose the weight of an edge $(6, 10)$ has changed, IncH2H first updates the weight of the affected shortcut $(8,10)$ in $S_G$ shown in Figure~\ref{fig:ch}(b). Then, starting from $\{(6,10),(8,10)\}$, it first updates the distance between vertices $\{6,8\}$ to their ancestor $10$ and iteratively identifies affected downward neighbors of $6$ and $8$ whose distance to $10$ has also changed. $L(2)$ will be updated as its distance to $10$ will be affected as well. For a distance query between vertices $1$ and $2$, $lca(1, 2)=8$ is first obtained, and then using the distances in $L(1)$ and $L(2)$ at the positions $[1, 2, 3]$ in $L(8)$, $d_G(1, 2) = 12$ is obtained according to Equation~\ref{lab:h2h_query}.
\end{example}

\begin{figure}
\centering
\includegraphics[width=0.5\textwidth]{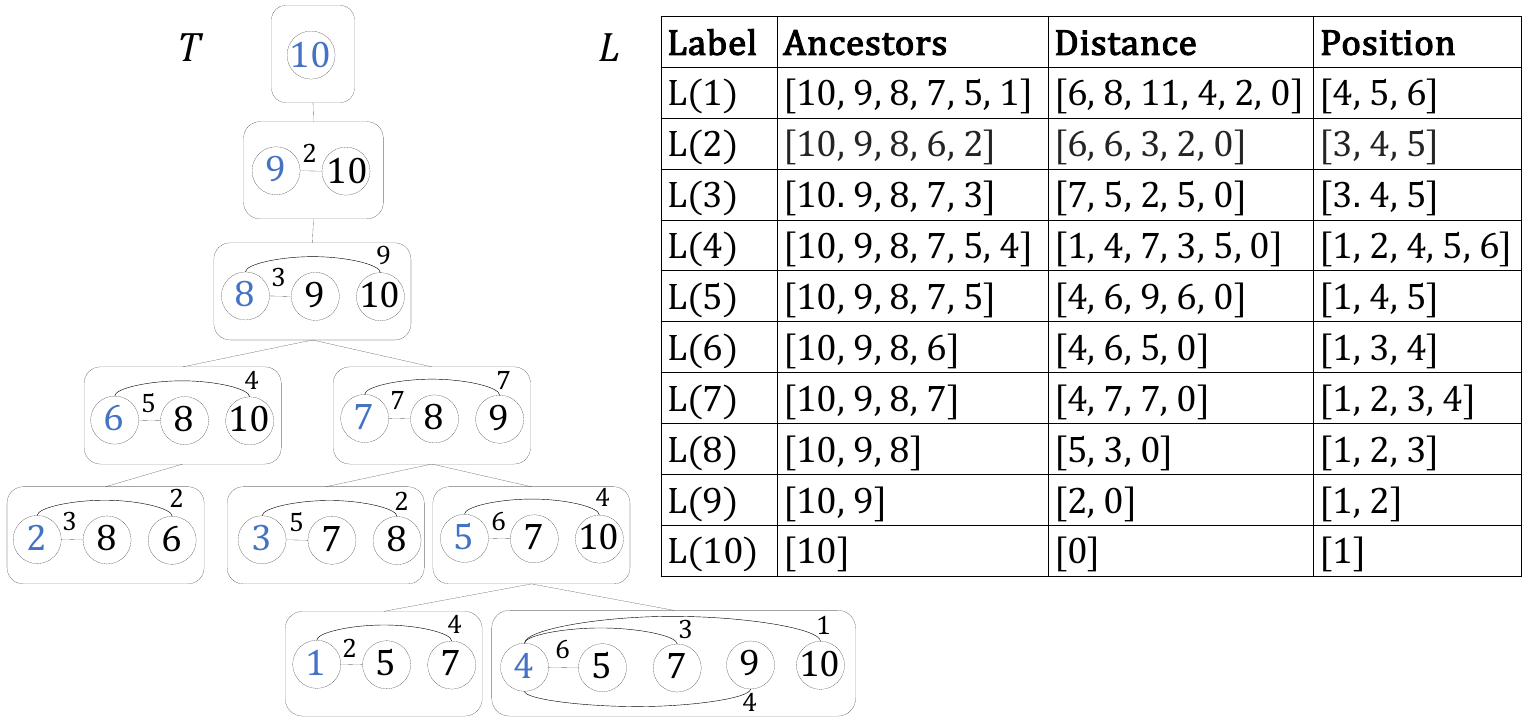}\vspace{-0.2cm}
\caption{A tree decomposition $T$ and its H2H-Index $L$.}
\label{fig:h2h}\vspace{-0.2cm}
\end{figure}

\subsection{Discussion}
Both state-of-the-art methods, DCH and IncH2H, exploit a vertex ordering to construct a shortcut graph $S_G$. This graph is then used in querying and maintenance by DCH, and to construct a tree decomposition by IncH2H. It is known that such a vertex ordering is crucial for constructing $S_G$ with a minimum number of shortcuts. However, finding an optimal vertex ordering that minimizes the number of shortcuts is challenging \cite{bauer2010preprocessing}. As a result, DCH suffers from scalability issues, leading to slow querying in large road networks. Similarly, IncH2H produces tree decompositions with very large height and width, resulting in huge \hh~sizes that hinder efficient maintenance. Additionally, IncH2H requires a complex mechanism for computing the least common ancestor of two vertices, which further degrades query performance.

\section{Our Proposed Solution}
Our proposed solution consists of three components: $(\langle \hq, \hu \rangle, L)$, where $\hq$ is a static \emph{query hierarchy}, $\hu$ is a dynamic \emph{update hierarchy}, and $L$ is a \emph{hierarchical labelling}. Further, $\hq$ and $\hu$ are \emph{order invariant} with respect to a vertex partial order $\ou$ over $V(G)$, and $L$ satisfies the 2-hop cover property.

In the following, we will discuss the details of our solution.

\subsection{Hierarchies: $\hq$ and $\hu$}

\noindent\textbf{Query Hierarchy.~}The purpose of $\hq$ is to improve query efficiency. A natural candidate for $\hq$ is a tree hierarchy, as tree-like structures are known to be very efficient for search~\cite{ouyang2018hierarchy,farhan2023hierarchical}. This is also evidenced by the state-of-the-art method IncH2H~\cite{zhang2022relative} where \hh~employs tree decomposition to accelerate query performance by restricting search to part of the labels of two given query vertices. However, as \hh~constructs a tree decomposition based on CH, it is vulnerable to produce very large heights $h$ and widths $w$ due to the ordering employed in CH. Consequently, query performance becomes slower as it needs to explore labels proportional to $w$ for answering a distance query.

\begin{definition}[Query Hierarchy]\label{def:td}
A \emph{query hierarchy} $\hq$ over a graph $G$ is a balanced binary tree, $\hq=(V_Q, E_Q, \omega_Q)$, where $V_Q$ is a set of tree nodes, $E_Q$ is a set of tree edges, and $\ell: V(G)\rightarrow V_Q$ is a total surjective function, satisfying the following properties:

\begin{enumerate}
\item Each  $N\in\mathcal{N}$ satisfies
\begin{equation*}
|\subtree(N_L)|, |\subtree(N_R)| \leq (1-\beta) \cdot |\subtree(N)|
\end{equation*}
where $0<\beta\leq 0.5$, $\subtree(N)$ denotes a subtree rooted at $N$, and $N_L$ and $N_R$ are the left and right children of $N$, respectively. 

\item Every path connecting $\{s,t\}\subseteq V(G)$ contains a vertex $a$ such that $\ell(a)$ is a common ancestor of $\ell(s)$ and $\ell(t)$.
\end{enumerate}

\end{definition}

\begin{example}
Figure~\ref{fig:query_update_hierarchy}(a) shows a query hierarchy $\hq$ of the road network $G$ in Figure~\ref{fig:ch}(a). Consider two vertices $6$ and $9$ in Figure~\ref{fig:ch}(a) and their corresponding tree nodes in Figure~\ref{fig:query_update_hierarchy}(a). The common ancestors of $6$ and $9$ are $\{\{3, 4, 10\}, \{2\}\}$ and there is a shortest-path between $6$ and $9$ i.e., $\langle 6, 10, 9 \rangle$ containing vertex $10$ from a common ancestor $\{3, 4, 10\}$. 
\end{example}

Following \cite{farhan2023hierarchical}, we construct a tree hierarchy for $\hq$ using the recursive bi-partitioning algorithm which finds balanced and minimal cuts to partition a graph. For each cut, we compute a partition bitstring and depth -- the number of vertices mapped to its ancestors in the tree hierarchy. Partition bitstrings allow us to find the lowest common ancestor of two nodes $s,t$ in constant time, and the depth of this ancestor is then used to identify which parts of $L(s)$ and $L(t)$ to use in computing their distance.
Compared to \cite{farhan2023hierarchical}, a query hierarchy requires that paths only need to intersect \emph{some} common ancestor rather than the lowest one, which thus allows omission of shortcuts during construction. This in turn leads to smaller cuts, reducing both the number of common ancestor vertices and overall labelling size.
However, we need to inspect all common ancestors to answer distance queries, not just the lowest one.

\vspace{0.1cm}
\noindent\textbf{Vertex Partial Order.~}~The tree structure of $\hq$ defines an ancestor-descendant relationship among tree nodes, which can be extended to a vertex partial order as defined below. 

\begin{definition}[Vertex Partial-Order $\ou$]
Let $\hq$ be a \htreename\ and $\preceq$ an arbitrary total order between vertices.
The \emph{vertex partial-order} $\ou$ induced by $\htree$ and $\preceq$ is defined as:
\begin{align*}
v \otree u \quad\Leftrightarrow\quad
& \text{$\ell(v)$ is a strict ancestor of $\ell(u)$ in $\hq$, or}\\
& \ell(v)=\ell(u) \text{ and } v \preceq u
\end{align*}
\end{definition}

We use $\anc(v)=\{w\in V(G)\mid w\ou v\}$ and $\desc(v)=\{w\in V(G)\mid v\ou w\}$ to denote the \emph{ancestors} and \emph{descendants} of $v$, respectively, and $lca(v,w)$ the common ancestor of vertices $v$ and $w$ which is largest w.r.t. $\ou$, i.e., their \emph{lowest common ancestor}.

\begin{example}\label{ex:vertex-partial-order}
Consider $\hq$ in Figure~\ref{fig:query_update_hierarchy}(a). By ordering the vertices within each tree node using their node identifiers, we obtain the vertex partial order
\[
\otree := \{ 3<4<10<\ldots, 1<5, 1<7, 2<6, 2<8<9 \}.
\]
We also have $anc(7)=\{3,4,10,1,7\}$, $anc(4)= \{3,4\}$, $lca(6,7)=10$, and $lca(6,9)=2$. 
\end{example}
\begin{figure}
\centering
\includegraphics[width=0.5\textwidth]{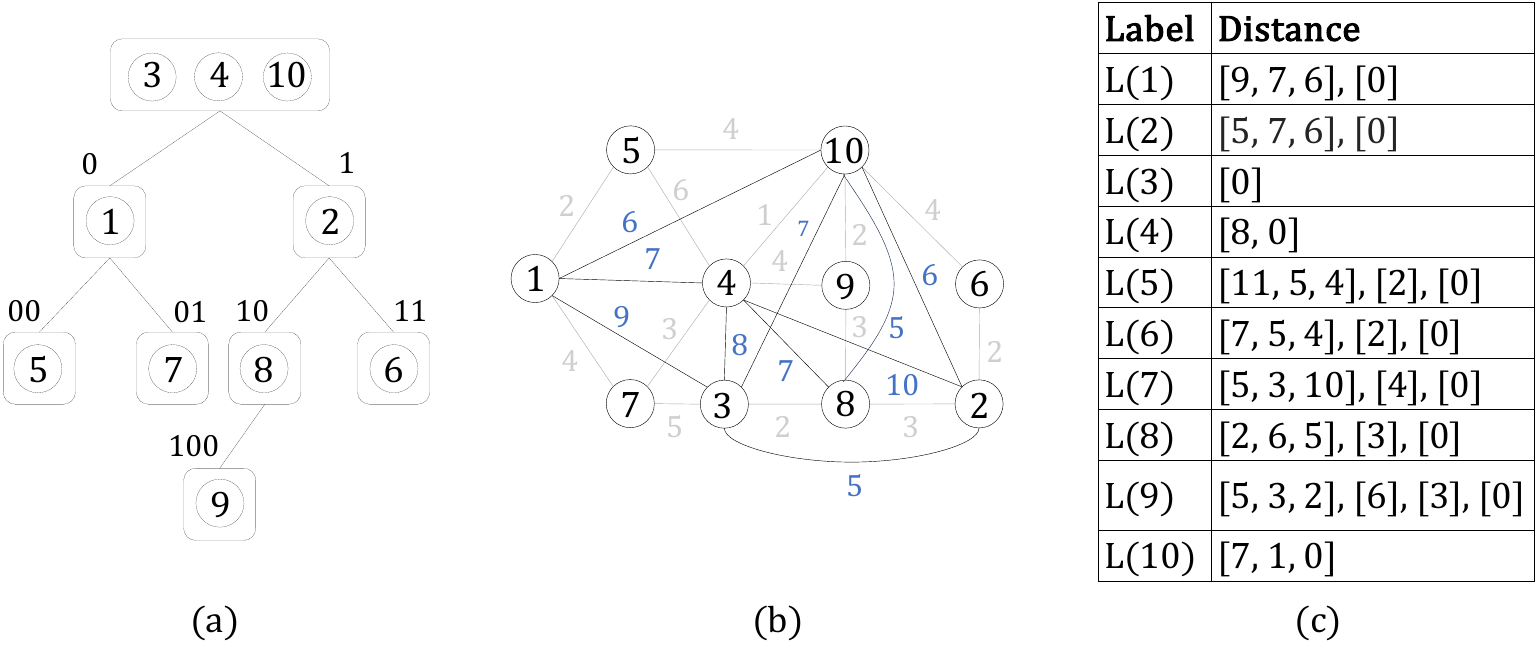}\vspace{-0.2cm}
\caption{(a) Query hierarchy $\hq$, (b) Update hierarchy $\hu$, and (c) Hierarchical labelling $L$.}
\label{fig:query_update_hierarchy}\vspace{-0.2cm}
\end{figure}

\vspace{0.1cm}
\noindent\textbf{Update Hierarchy.~}Unlike $\hq$, the purpose of $\hu$ is to support efficient maintenance of the labelling under dynamic changes. Thus, instead of tree-like structures, a shortcut graph (i.e., the original graph $G$ plus a set of shortcuts) from CH \cite{geisberger2008contraction} would be a good candidate for $\hu$, which can not only preserve distances between vertices in $G$ but also enable fast search for updating shortcuts and distance labels being affected. A central concept of CH is valley path, where intermediate vertices have higher ranks than their endpoints. We extend this notion to partial orders in our work.

\begin{definition}[Valley Path]
A path $p$ between $v$ and $w$ with $w\otree v$ is a \emph{valley path} iff $v\otree u$ for all $u\in V(p)\setminus\{v,w\}$.
\end{definition}

These valley paths form the basis of our update hierarchy.

\begin{definition}[Update Hierarchy]
An update hierarchy $\hu$ over a graph $G$ and a vertex partial order $\otree$ is a shortcut graph which contains a shortcut $(v,w)$ for every valley path between $v$ and $w$.
The weight associated with a shortcut $(v,w)$ is the length of the shortest valley path between $v$ and $w$.
\end{definition}

\begin{example}
Figure~\ref{fig:query_update_hierarchy}(b) illustrates update hierarchy $\hu$ over an example road network $G$ shown in Figure~\ref{fig:ch}(a) using the vertex partial order from Example~\ref{ex:vertex-partial-order}.
There is a shortcut edge $(1,4)$ in $\hu$ because valley paths $\langle1, 7, 4\rangle$ and $\langle1, 5, 4\rangle$ between $1$ and $4$ exist in $G$.
Note that $\langle1,5,10,4\rangle$ is not a valley path since $10\otree 1$.
We have $\omega(1,4)=7$, the weight of the shortest valley path $\langle1, 7, 4\rangle$.
\end{example}

Let $G=(V,E)$ be a road network, $\Delta(E)$ be a set of edge weight updates on $G$, and $G\oslash\Delta(E)$ the resulting graph after applying $\Delta(E)$ to $G$.
We also use $\hu=(V_U, E_U, \omega_U)$ and $\hu^{\Delta}=(V_U^{\Delta}, E_U^{\Delta}, \omega_U^{\Delta})$ to refer to the update hierarchies corresponding to $G$ and $G\oslash\Delta(E)$, respectively. The following properties  are satisfied by $\hu$ and $\hu^{\Delta}$ :

\begin{itemize}
    \item[(U1).] {\textbf{Structural stability}}: $\hu$ and $\hu^{\Delta}$ differ only in edge weights, i.e., $(V_U, E_U)=(V_U^{\Delta}, E_U^{\Delta})$.
    \item[(U2).] {\textbf{Bounded searching}}: A weight update of $(v,w)\in\Delta(E)$ only affects shortcuts $(v',w')$ with $v',w'\otree v,w$. 
\end{itemize}

We construct $\hu$ using the approach \cite{ouyang2020efficient}. However, unlike \cite{wei2020architecture,ouyang2020efficient} which use a pre-defined total vertex order, we perform vertex contractions following $\ou$ in decreasing order. 
The following lemma shows that the \chname\ based on $\otree$ is identical to the classic CH based on any total order extending $\otree$. 

\begin{lemma}\label{L:ch-partial}
Let $\preceq$ be a total ordering on $V(G)$ extending $\otree$, $w\preceq v$, and $p$ a valley path between $v$ and $w$ w.r.t. $\preceq$.
Then $w\otree v$ and $p$ is also a valley path w.r.t. $\otree$.
\end{lemma}

\begin{proof}
Definition~\ref{def:td} implies that $p$ must contain a common ancestor $r$ with $r\otree v,w$.
As $p$ is a valley path w.r.t. $\preceq$, we must have $w\preceq r$. Since $\preceq$ extends $\otree$, it follows that $r=w$.
This shows $w\otree v$. 
Consider now any intermediate vertex $u\in V(p)\setminus\{v,w\}$ and the subpath $p'$ of $p$ connecting $v$ and $u$.
Again, Definition~\ref{def:td} implies the existence of $r\in V(p')$ with $r\otree v$ and $r\otree u$. Since $p$ is a valley path w.r.t. $\preceq$, we must have $v\preceq r$.
It follows that $v=r\otree u$, which makes $p$ a valley path w.r.t. $\otree$.
\end{proof}

While our update hierarchy is essentially a contraction hierarchy, there are some subtle differences. To maintain the roles of upward and downward neighbors in our context, we define
\begin{align*}
\nup(v) &= \{ u \mid (v,u)\in E(\hu) \land u \otree v \}; \\
\ndown(v) &= \{ u \mid (v,u)\in E(\hu) \land v \otree u \}.
\end{align*}

\subsection{Hierarchical Labelling: $L$}
We present the design of the hierarchical labelling $L$, which is based on \emph{the principle of separation} between $\hq$ and $\hu$. This ensures that distance queries can be answered by searching $L$ via $\hq$, while dynamic changes in a road network can be reflected in $L$ via $\hu$.

Conceptually, the hierarchical labelling $L$ in our work is associated with a \emph{distance scheme} $\Gamma_L$ and a \emph{distance map} $\gamma_L$ defined over $\Gamma_L$. Below, we provide detailed definitions of $\Gamma$ and $\gamma$, omitting the subscript $L$ to simplify the notation.

\vspace{0.1cm}
\noindent\textbf{Distance Scheme.~}The distance scheme $\Gamma$ describes how distances are to be stored in $L$, which is determined by query hierarchy $\hq$. Since each vertex $v\in V(G)$ is associated with a set of ancestor tree nodes $\mathcal{N}(v)=\{\ell(w)\mid w\in anc(v)\}$ in $\hq$, we start with defining the notion of \emph{tree node scheme} to describe the positions of vertices that are associated with each tree node.

\begin{definition}[Tree Node Scheme] Let $N^t\in \mathcal{N}(v)$, $depth(N^t)$ be the depth of $N^t$ in $\hq$, and $t=depth(N^t)$. The \emph{tree node scheme} of $v$ at $N^t$ is $\Gamma^{t}(v)=[w_1,\dots, w_k]$ where   $w_1\ou\dots\ou w_k$, and $\{w_1,\dots, w_k\}=\{w_i|w_i\in \ell^{-1}(N^t), w_i\ou v\}$. 
\end{definition}

Based on tree node schemes for ancestor tree nodes, we define the distance scheme $\Gamma$ for each vertex $v\in V(G)$.

\begin{definition}[Distance Scheme]\label{def:scheme}
  Let $N=\ell(v)$.
 Then, the \emph{distance scheme} $\Gamma(v)$ of $v$ is defined as
\begin{equation}
  \Gamma(v)=\Gamma^{0}(v)\hspace{0.1cm}|\hspace{0.1cm}\Gamma^{1}(v)\hspace{0.1cm}|\hspace{0.1cm}\dots\hspace{0.1cm}|\hspace{0.1cm}\Gamma^{m}(v), 
\end{equation}
where $\Gamma^{t}(v)$ is the tree node scheme of $v$ at a tree node $N^t\in \mathcal{N}(v)$ and $m=depth(N)$.
\end{definition}

Our distance scheme is purely conceptual (no data is actually stored) and does not change as edge weights in $G$ are modified.

\begin{algorithm}[t]
\caption{Hierarchical Labelling Construction}\label{algo:label-construct}
\SetCommentSty{textit}
\SetKwFunction{FMain}{\om}
\SetKwProg{Fn}{Function}{}{end}
\SetKw{and}{and}
\Fn{\FMain{$\tau$, $\hu$}}{
    initialize $L_*[*]=\infty$\\
    \tcp{copy shortcuts}
    \ForEach{$(v,w)\in E(\hu)$ with $\lidx{v} > \lidx{w}$}{
        $L_v[w]\gets \omega(v,w)$
    }
    \tcp{compute label distances top-down}
    \ForEach{$v\in V$ in increasing order of $\lidx{v}$}{
        \ForEach{$w\in \nup(v)$}{
            \ForEach{$i\in [0, \lidx{w}]$ \label{L:ch-construct-iter-anc}}{
                $L_v[i]\gets \min(L_v[i],\; \omega(v,w) + L_w[i])$
            }
        }
    }
}
\end{algorithm}

\vspace{0.1cm}
\noindent\textbf{Distance Map.~}The distance map $\gamma$ stores distance entries into tree node schemes of the distance scheme $\Gamma$. Note that the distance entries we store are not necessarily distances in $G$ as for IncH2H or HC2L, but distances within a subgraph of $\hu$.

\begin{definition}[Subgraph Distance]\label{def:subgraph-distance}
We denote by $\sgd$ the distance between $w$ and $v$ in the subgraph of $\ch$ induced by 
\[
\desc(w)\cap\anc(v) = \{ u \mid w \otree u \otree v \}
\]
\end{definition}

We use $L_v[w]$ to denote the distance between $v$ and $w$ in this subgraph, i.e., $L_v[w] = \sgd(v, w)$. As we shall see later, these distance entries can also be expressed as distances in certain subgraphs of $G$. This close relationship to $\ch$ enables us to update $L$ efficiently based on changes to $\ch$.

\begin{definition}[Distance Map]\label{def:map}
A \emph{distance map} is a function $\gamma: V(G)\times V(G) \to \mathbb{R}_{\geq 0}$ such that, for any $\Gamma^{t}(v)=[w_1,\dots, w_k]$, 
\begin{equation}
  L^t(v)=\gamma(\Gamma^{t}(v))=[\gamma(v,w_1),\dots, \gamma(v,w_k)]  
\end{equation}
 where $\gamma(v,w_i)=L_v[w_i]$ for $i=1,\dots, k$. Let $N=\ell(v)$. Accordingly, $\gamma$ over $\Gamma(v)$ defines the label $L(v)$ as 
\begin{equation}
  L(v)=L^{0}(v)\hspace{0.1cm}|\hspace{0.1cm}L^{1}(v)\hspace{0.1cm}|\hspace{0.1cm}\dots\hspace{0.1cm}|\hspace{0.1cm}L^{m}(v). 
\end{equation}
\end{definition}

\begin{example}
Consider Figure~\ref{fig:query_update_hierarchy}. The distance scheme of vertex $7$ is $\Gamma(7)=\Gamma^0(7)\mid \Gamma^1(7)\mid \Gamma^2(7)$ with $\Gamma^0(7)=\{3,4,10\}, \Gamma^1(7)=\{1\}$ and $\Gamma^2(7)=\{7\}$ at depth $0, 1$ and $2$ in Figure~\ref{fig:query_update_hierarchy}(a). Accordingly, the distance map of vertex $7$ is shown as $L(7)$ in Figure~\ref{fig:query_update_hierarchy}(c) storing distances to ancestors in $\Gamma(7)$. $L_7[\tau(10)=2]$ is 10 rather than 4, the distance in the subgraph of $\hu$ induced by $\{10,1,7\}$.
\end{example}

We shall denote by $\tau(v)=|\{ w\in V(G)|w \ou v \}|$ the number of ancestors of vertex $v$ w.r.t. $\otree$. In our implementation we use $\tau(u)\leq\tau(v)$ to check whether $u\otree v$ holds. By Lemma~\ref{L:ch-partial} shortcut edge endpoints are always comparable w.r.t. $\otree$, so this characterization is sound. Note that $\tau$ plays the same role as the vertex ordering $\pi$ used in contraction hierarchies, except that the ordering is reversed and vertices incomparable w.r.t. $\otree$ may share $\tau$ values.

Our bottom-up approach to construct the hierarchical labelling $L$ is described as Algorithm~\ref{algo:label-construct}. We compute the label $L(v)$ for each vertex $v$ in increasing order of $\tau(v)$. Specifically, we inspect ancestors of upward neighbors $w\in \nup(v)$ to compute $L_v[i]$ (Lines 7-9).

\subsection{Queries and Updates}\label{subsec:queries-updates}
In the following, we discuss how distance queries and weight updates are processed in our work. 

\vspace{0.1cm}
\noindent\emph{\underline{Distance Queries}.~}Given two vertices $s,t\in V(G)$, a distance query between $s$ and $t$ is processed in two phases via $\hq$ and $L$: 
\begin{itemize}
    \item[(1)] The common ancestors $\anc(s)\cap\anc(t)$ of $s$ and $t$ are quickly found through $lca(s, t)$ in $\hq$;
    \item[(2)] By restricting the search of $L(s)$ and $L(t)$ to only vertices occurring in $\anc(s)\cap\anc(t)$, $d_G(s,t)$ is defined as
\begin{align*}
d_G(s,t) = \min\{ L_s[r]+L_t[r] \mid\; &L_s[r]\in {L}(s), L_t[r]\in {L}(t),\nonumber\\
& r\in\anc(s)\cap\anc(t) \}. 
\end{align*}
\end{itemize}

As we identify tree nodes in $\hq$ via bitstrings, we can compute the level $l$ of $lca(s,t)$ as the common prefix of the bitstrings of $s$ and $t$. Then the common ancestors of $s$ and $t$ are the vertices at the depths between $0$ and $l$ which can also be efficiently found.

\begin{example}
Consider a query pair $(6, 9)$. The bitstrings of vertices $6$ and $9$ are 11 and 100, respectively, shown in Figure~\ref{fig:query_update_hierarchy}(a). The depth of $lca(6, 9) = \{2\}$ is $1$ and their common ancestors are at the depths $0 \leq i \leq 1$, i.e., $anc(6)\cap anc(9)=\{3, 4, 10, 2\}$.
\end{example}

\vspace{0.1cm}
\noindent\emph{\underline{Weight Updates}.~}Given any weight updates $\Delta(E)$ on $G$, $\hu$ and $L$ are processed to reflect $\Delta(E)$, which also involves two phases:
\begin{itemize}
    \item[(1)] $\hu$ is processed to find all shortcuts $\Delta(S)$ on $\hu$ affected by $\Delta(E)$ and the weights of these shortcuts are updated.
    \item[(2)] The distance entries of $L$ affected by $\Delta(S)$ are updated.
\end{itemize}

Observe that using our definitions of upward and downward neighbors, Property~\ref{lab:weight_property} holds for $\hu$ as well, and can be used to maintain shortcuts.
Let $\omega_U$ and $\omega_U^{\Delta}$ refer to the shortcut weights w.r.t. $G$ and $G\oslash\Delta(E)$, respectively.

\begin{definition}[Affected Shortcut]
A shortcut $(v,w)$ in $\hu$ is \emph{affected} by $\Delta(E)$ iff $\omega_U(v,w)\neq\omega_U^\Delta(v,w)$.
\end{definition}

\begin{lemma}
If $(v,w)\in E(\hu)$ is affected by $\Delta(E)$, then either
\begin{enumerate}
\item $(v,w)\in\Delta(E)$, or
\item $(v,x)$ or $(x,w)$ is affected for some $x\in \ndown(v)\cap \ndown(w)$.
\end{enumerate}
\end{lemma}

This enables maintenance of $\hu$ by focusing on the intermediate vertex $x$ of the shortcut triangle $v{-}x{-}w$.
A bottom-up search (w.r.t. $\otree$) ensures that the weights of $(v,x)$ and $(x,w)$ have already been updated before using $(v,w)$.
We describe this, and how to update $L$ based on $\hu$, in the next section.

\begin{algorithm}[t]
\caption{$\ch$ Maintenance (Decrease)}
\label{algo:dch-minus}
\SetCommentSty{textit}
\SetKwFunction{FMain}{$\dch^-$}
\SetKwProg{Fn}{Function}{}{end}
\SetKw{Let}{let}
\Fn{\FMain{$\tau$, $\hu$, $\Delta(E)$}}
{
    $\mathcal{Q}\gets\emptyset$ \tcp{a priority queue}
    \ForEach{$(v,w,\omega_{new})\in\Delta(E)$ with $\lidx{v}>\lidx{w}$}
    {
        \If{$\omega(v,w) > \omega_{new}$}
        {
            $\omega(v,w) \gets \omega_{new}$\\
            add $(v,w)$ into $\mathcal{Q}$
        }
    }
    \ForEach{$(v,w)\in\mathcal{Q}$ in decreasing order of $\lidx{v}$}
    {
        \ForEach{$w'\in\nup(v)\setminus\{w\}$}
        {
            \If{$\omega(w,w')>\omega(v,w) + \omega(v,w')$}
            {
                $\omega(w,w')\gets\omega(v,w) + \omega(v,w')$\\
                add $(\max_\tau(w,w'), \min_\tau(w,w'))$ into $\mathcal{Q}$
            }
        }
    }
    \Return affected shortcuts
}
\end{algorithm} 

\begin{algorithm}[t]
\caption{$\ch$ Maintenance (Increase)}
\label{algo:dch-plus}
\SetCommentSty{textit}
\SetKwFunction{FMain}{$\dch^+$}
\SetKwProg{Fn}{Function}{}{end}
\SetKw{Let}{let}
\Fn{\FMain{$\tau$, $\hu$, $\Delta(E)$}}{
    $\mathcal{Q}\gets\emptyset$ \tcp{a priority queue}
    \ForEach{$(v,w,\omega_{old})\in \Delta(E)$ with $\lidx{v}>\lidx{w}$}{
        \If{$\omega(v,w) = \omega_{old}$}{
            add $(v,w)$ into $\mathcal{Q}$
        }
    }
    \ForEach{$(v,w)\in\mathcal{Q}$ in decreasing order of $\lidx{v}$}{
        \tcp{recompute shortcut distance}
        $\omega_{new}\gets\begin{cases}
        \omega & \text{if }(v,w,\omega)\in E(G)\oslash\Delta(E)\\
        \infty & \text{otherwise}
        \end{cases}$\\
        \ForEach{$x\in \ndown(v)\cap \ndown(w)$}{
            $\omega_{new}\gets\min(\omega_{new},\; \omega(x,v) + \omega(x,w))$
        }
        \If{$\omega(v,w)\neq\omega_{new}$}{
            \ForEach{$w'\in \nup(v)\setminus\{w\}$}{
                \If{$\omega(w,w')=\omega(v,w) + \omega(v,w')$}{
                    add $(\max_\tau(w,w'), \min_\tau(w,w'))$ into $\mathcal{Q}$
                }
            }
            $\omega(v,w)\gets\omega_{new}$
        }
    }
    \Return affected shortcuts   
}
\end{algorithm}

\section{Dynamic Algorithms}\label{sec:dynamic-algorithm}
In this section we propose two dynamic algorithms, \om$^-$ and \om$^+$, to efficiently maintain the hierarchical labelling $L$, handling edge weight decrease and increase, respectively. Since $L$ is constructed upon the update hierarchy $H_U$, we first maintain $H_U$ and then maintain $L$ using $\hu$.
Note that each $L_v[w]$ stores the distance between two vertices $v$ and $w$ in a (fairly small) induced subgraph of $\hu$, and not the distance in $G$, which makes some update operations simpler in our work, compared to IncH2H.

\begin{algorithm}[t]
\caption{\om\ Maintenance (Decrease)}
\label{algo:ch-update-dec}
\SetCommentSty{textit}
\SetKwFunction{FMain}{\om$^-$}
\SetKwProg{Fn}{Function}{}{end}
\SetKw{Let}{let}
\Fn{\FMain{$\tau$, $H_U$, $L$, $\Delta(E)$}}{
    $\Delta(S)\gets$ $\dch^-$($\tau$, $H_U$, $\Delta(E)$) \tcp{shortcut updates}
    $\mathcal{Q}\gets\emptyset$ \tcp{a priority queue}
    \tcp{update distances involving ancestors}
    \ForEach{$(v,w,\omega_{new})\in \Delta(S)$ with $\omega_{new} < L_v[w]$}{
        \ForEach{$i\in[0,\lidx{w}]$}{
            \If{$\omega_{new} + L_w[i] < L_v[i]$}{
                $L_v[i]\gets \omega_{new} + L_w[i]$\\
                add $(v,i)$ into $\mathcal{Q}$
            }
        }
    }
    \tcp{identify and update distances involving descendants}
    \ForEach{$(v,i)\in\mathcal{Q}$ in increasing order of $\lidx{v}$}{
        \ForEach{$u\in \ndown(v)$}{
            \If{$L_u[v] + L_v[i] < L_u[i]$}{
                $L_u[i]\gets L_u[v] + L_v[i]$\\
                add $(u,i)$ into $\mathcal{Q}$
            }
        }
    }
}
\end{algorithm}
\begin{algorithm}[t]
\caption{\om\ Maintenance (Increase)}
\label{algo:ch-update-inc}
\SetCommentSty{textit}
\SetKwFunction{FMain}{\om$^+$}
\SetKwProg{Fn}{Function}{}{end}
\SetKw{Let}{let}
\Fn{\FMain{$\tau$, $\hu$, $L$, $\Delta(E)$}}{
    $\Delta(S)\gets\dch^+$($\tau$, $\hu$, $\Delta(E)$) \tcp{shortcut updates} 
    $\mathcal{Q}\gets\emptyset$ \tcp{a priority queue}
    \tcp{identify distances to ancestors to update}
    \ForEach{$(v,w,\omega_{old})\in\Delta(S)$ with $\omega_{old} = L_v[w]$}{
        \ForEach{$i\in[0,\lidx{w}]$}{
            \If{$\omega_{old} + L_w[i] = L_v[i]$}{
                add $(v,i)$ into $\mathcal{Q}$
            }
        }
    }
    \tcp{identify and update distances involving descendants}
    \ForEach{$(v,i)\in\mathcal{Q}$ in increasing order of $\lidx{v}$}{
        $\omega_{new}\gets\infty$ \tcp{new distance from $v$ to $i$}
        \ForEach{$w\in \nup(v)$ with $\lidx{w}\geq i$}{
            $\omega_{new}\gets \min(\omega_{new},\; \omega(v,w) + L_w[i])$
        }
        \tcp{distance may not have changed after all}
        \If{$\omega_{new} > L_v[i]$}{
            \ForEach{$u\in\ndown(v)$}{
                \If{$L_u[v] + L_v[i] = L_u[i]$}{
                    add $(u,i)$ into $\mathcal{Q}$
                }
            }
            $L_v[i]\gets\omega_{new}$
        }
    }
}
\end{algorithm}

\subsection{Edge Weight Decrease}
Algorithm~\ref{algo:ch-update-dec} shows \emph{\om$^-$}. It first maintains the update hierarchy $\hu$ using $\dch^-$ described as Algorithm~\ref{algo:dch-minus}.
Then, in Algorithm~\ref{algo:dch-minus}, for each update $(v,w,\omega_{new})\in\Delta(E)$ with $\tau(v)>\tau(w)$, $\dch^-$ tests whether the old weight in $\hu$ is greater than $\omega_{new}$.
If so, it updates the weight of the edge $(v, w)$ in $\hu$ and pushes $(v,w)$ to a priority queue $\mathcal{Q}$ (Lines 3-8).
Each affected shortcut $(v,w)\in\mathcal{Q}$ is processed in the decreasing  order of $\tau(v)$, i.e., for each $w'\in \nup(v)$ it checks whether the weight of the shortcut $(w, w')$ in $\hu$ is greater than the weight of the new path through $v$ (Lines 7-9), updates it accordingly and pushes $(w, w')$ to $\mathcal{Q}$ (Lines 10-11). The maintenance of $\hu$ continues until all affected shortcuts are updated and finally affected shortcuts $\Delta(S)$ are returned to Algorithm~\ref{algo:ch-update-dec} (Line 2).

At Lines 4-8 of Algorithm~\ref{algo:ch-update-dec}, for each shortcut $(v,w,\omega_{new})\in\Delta(S)$, it examines distance entries in the label of $v$ w.r.t. ancestors of $w$.
Specifically, for each ancestor $i\in[0,\tau(w)]$, if the old distance $L_v[i]$ is greater than the sum of the new distance $\omega_{new}$ between $v$ and $w$ and $L_w[i]$, it updates $L_v[i]$ and adds $(v, i)$ to a priority queue $\mathcal{Q}$, as the distance update between $v$ and ancestor $i$ may have affected the labels of descendants of $v$.
Next, at Lines 9-13, it processes pairs $(v,i)\in\mathcal{Q}$ in increasing order of $\tau(v)$ to ensure the availability of correct labels of ancestors when updating the labels of descendants. That is, for each pair $(v,i)\in\mathcal{Q}$, it examines each $u\in\ndown(v)$ to check whether distance between $u$ to $i$ is changed via ancestor $v$ ($L_v[i]$ which has been updated) and accordingly update and enqueue pair $(u,i)$ into $\mathcal{Q}$. This process continues iteratively until no further labels are updated. We note that this approach does not consider all paths of Definition~\ref{def:subgraph-distance}, but only paths in which vertices are ordered by $\tau$.
This will be justified later by Lemma~\ref{L:shortcut-chain}.

\subsection{Edge Weight Increase}
Algorithm~\ref{algo:ch-update-inc} describes \om$^+$ for edge weight increase. Similar to \om$^-$, it first maintains $\hu$ using $\dch^+$ described as Algorithm~\ref{algo:dch-plus}.
In Algorithm~\ref{algo:dch-plus}, for each update $(v,w,\omega_{new})\in\Delta(E)$ with $\tau(v)>\tau(w)$, $\dch^+$ tracks path(s) through the updated edge and pushes $(v,w)$ to a priority queue $\mathcal{Q}$ (Lines 3-5).
Then, each affected shortcut $(v,w)\in\mathcal{Q}$ is processed in the decreasing order of $\tau(v)$.
Different from the weight decrease case, the updated weight $\omega_{new}$ of shortcut $(v,w)$ is obtained using Equation~\ref{eq:weight_property} (Lines 7-9) and whether the weight obtained is different from the old weight is checked (Line 10).
If so, it examines upward neighbors $w'\in \nup(v)$ to identify potentially affected shortcuts $(w, w')$ and pushes them to $\mathcal{Q}$, before updating $(v,w)$ (Lines 11-14).
The maintenance of $\hu$ continues until all affected shortcuts are updated, and finally affected shortcuts $\Delta(S)$ are returned to Algorithm~\ref{algo:ch-update-inc} (Line 2).

Then, at Lines 4-7, for each affected shortcut $(v,w,\omega_{old})\in\Delta(S)$, it examines distance entries in the label of $v$ w.r.t. ancestors of $w$.
Those with the same distance as a path ending in the shortcut $(v,w)$ are potentially affected, and added to a priority queue $\mathcal{Q}$.
Next, at Lines 8-16, it processes pairs $(v,i)\in\mathcal{Q}$ in increasing order of $\tau(v)$.
For each pair $(v,i)\in\mathcal{Q}$, it computes the new distance $\omega_{new}$ for $L_v[i]$ (Lines 10-11). If $\omega_{new}$ is greater than the old distance $L_v[i]$, it examines each $u\in\ndown(v)$ to check whether a shortest path between $u$ and $i$ passes through vertex $v$ (Line 14) and might thus be affected and if so enqueues $(u,i)$.
Finally, we update the label entry $L_v[i]$ (Line 16). This process continues iteratively until no further labels are updated. Note that unlike IncH2H, $\dch^+$ and \om$^-$ do not track the support for shortcuts or labels as in practice most of them have a support of one. This helps to save space at the cost of recomputing distances of some unaffected shortcuts and label entries. Experimentally, we found the fraction of unnecessary recomputations to be small enough to justify this trade-off.

\begin{example}
Consider that the weight of edge $(7,4)$ is decreased to $1$ from $3$ in the road network depicted in Figure~\ref{fig:ch}(a). As the weight is decreased, $\omega(7,4)$ of the shortcut $(7,4)$ is updated in $\hu$ and inserted into $\mathcal{Q}$ (Lines 4-6). Lines 7-11 then process $(7,4)$ and inspect weights of the shortcuts $\{(1,4),(4,3)\}$. As $\omega(1,4) > \omega(7,1)+\omega(7,4)$ and $\omega(4,3) > \omega(7,3)+\omega(7,4)$, there weights are updated to $5$ and $6$, respectively. Then in Algorithm~\ref{algo:ch-update-dec}, labels of affected vertices are updated w.r.t. $\Delta(S)=\{(7,4,1),(1,4,5),(4,3,6)\}$. The label entries $L_7[\tau[4]=1], L_1[\tau[4]=1]$ and $L_4[\tau[3]=0]$ are updated to $1, 5$ and $6$, respectively, and inserted into $\mathcal{Q}$ at Lines 4-8.
Due to the updates of $(1,4)$ and $(4,3)$, the downward neighbors $\{5,7\}\in\ndown(1)$ and $\{1,2,5,7,8,9,10\}\in\ndown(4)$ are inspected in Lines 9-13 to see if their distances to 4 and 3 change, but no more labels are updated.

Now consider the weight of the edge $(4,7)$ increased to $5$ from $3$ in the road network depicted in Figure~\ref{fig:ch}(a). Again, the shortcuts $\{(1,4),(4,3)\}$ are inspected w.r.t. $(7,4)$ at Lines 11-13 and found to be affected as $\omega(1,4)=\omega(7,1)+\omega(7,4)$ and $\omega(4,3)=\omega(7,4)+\omega(7,3)$. The weights of $\Delta(S)=\{(7,4),(1,4),(4,3)\}$ are updated in $\hu$ at Line 14 (Algorithm~\ref{algo:dch-plus}) and returned to Algorithm~\ref{algo:ch-update-inc} (Line 2). Then, in Algorithm~\ref{algo:ch-update-inc}, after updating $L_7[\tau[4]=1], L_1[\tau[4]=1]$ and $L_4[\tau[3]=0]$ there are no further updates of downward neighbors.
\end{example}

\subsection{Parallel Maintenance}
We now introduce parallel variants for our dynamic algorithms \om$^-$ and \om$^+$, namely \om$^-_p$ and \om$^+_p$  described as Algorithm~\ref{algo:par-dec}-\ref{algo:par-inc}, respectively. We observe that the most time consuming part of Algorithms~\ref{algo:ch-update-dec} and~\ref{algo:ch-update-inc} is to compute updated distances to descendants at Lines~(9-13) and (8-16) respectively. We can parallelize these parts to significantly improve performance. The main idea is to distribute this work amongst multiple threads. We group entries $(v,i)\in\mathcal{Q}$ w.r.t. ancestors $i$ in their respective queues $Q_i$ (Line 3). This way all entries with the same $i$ are processed within the same thread. To make this work without costly synchronization operations, we must ensure that distance labels accessed by different threads are distinct. Observe that all entries generated when processing an entry $(v,i)\in Q_i$ are in the form $(u,i)$, and thus exclusive to the processing thread. Furthermore, the only labels used for comparison are $L_u[v]$ and $L_v[i]$. The latter also ``belongs'' to the processing thread and will not be updated externally, but $L_u[v]$ could be updated by the thread processing $\lidx{v}$, which is problematic. However, we can simply replace $L_u[v]$ with $\omega(u,v)$ which is not updated by any thread, and still obtain correct updates (by Lemma~\ref{L:shortcut-chain}).

\begin{algorithm}[t]
\caption{\om\ Parallel Maintenance (Decrease)}
\label{algo:par-dec}
\SetCommentSty{textit}
\SetKwFunction{FMain}{\om$^-_p$}
\SetKwProg{Fn}{Function}{}{end}
\SetKw{Let}{let}
\Fn{\FMain{$\tau$, $\hu$, $L$, $\Delta(E)$}}{
    \dots\\
    \tcp{identify and update distances involving descendants}
    partition $\mathcal{Q}$ into $\mathcal{Q}_0, \mathcal{Q}_1, \dots$ with $(v,i)\in \mathcal{Q}_i$\\
    \ForEach{$\mathcal{Q}_i$ in parallel}{
        \ForEach{$(v,i)\in \mathcal{Q}_i$ in increasing order of $\lidx{v}$}{
            \ForEach{$u\in \ndown(v)$}{
                \If{$\omega(u,v) + L_v[i] < L_u[i]$}{
                    $L_u[i]\gets \omega(u,v) + L_v[i]$\\
                    add $(u,i)$ into $\mathcal{Q}_i$
                }
            }
        }
    }
}
\end{algorithm}

\begin{algorithm}[t]
\caption{\om\ Parallel Maintenance (Increase)}
\label{algo:par-inc}
\SetCommentSty{textit}
\SetKwFunction{FMain}{\om$^+_p$}
\SetKwProg{Fn}{Function}{}{end}
\SetKw{Let}{let}
\Fn{\FMain{$\tau$, $\hu$, $L$, $\Delta(E)$}}{
    \dots\\
    \tcp{identify and update distances involving descendants}
    partition $\mathcal{Q}$ into $\mathcal{Q}_0, \mathcal{Q}_1, \dots$ with $(v,i)\in \mathcal{Q}_i$\\
    \ForEach{$\mathcal{Q}_i$ in parallel}{
        \ForEach{$(v,i)\in \mathcal{Q}_i$ in increasing order of $\lidx{v}$}{
            $\omega_{new}\gets\infty$ \tcp{new distance from $v$ to $i$}
            \ForEach{$w\in \nup(v)$ with $\lidx{w}\geq i$}{
                $\omega_{new}\gets \min(\omega_{new},\; \omega(v,w) + L_w[i])$
            }
            \tcp{distance may not have changed after all}
            \If{$\omega_{new} > L_v[i]$}{
                \ForEach{$u\in\ndown(v)$}{
                    \If{$\omega(u,v) + L_v[i] = L_u[i]$}{
                        add $(u,i)$ into $\mathcal{Q}_i$
                    }
                }
                $L_v[i]\gets\omega_{new}$
            }
        }
    }
}
\end{algorithm}
\section{Theoretical Results}\label{S:theory}
We discuss several theoretical aspects of our solution, including correctness, 2-hop cover property, and complexity.

\vspace{0.1cm}
\noindent\textbf{Correctness Analysis.~}~We first establish a connection between shortcuts in $\hu$ and distance entries in $L$.

\begin{definition}[Shortcut Chain]\label{def:shortcut-chain}
A \emph{shortcut chain} is a descending sequence of vertices $v_1 \ou \ldots \ou v_n$ such that $v_i$ and $v_{i+1}$ (for $i\in[1,n)$) are connected by a shortcut in $\hu$, i.e., the shortcut chain \emph{connects} $v_1$ and $v_n$.
The \emph{length} of a shortcut chain is the sum of weights associated with the shortcuts involved.
\end{definition}

\begin{example}
Consider the update hierarchy $\hu$ from Figure~\ref{fig:query_update_hierarchy}(b).
Here 3-2-6 forms a shortcut chain of length 7, connecting 3 and 6. 
\end{example}

Let $d_G^w(v,u)$ denote the distance between two vertices $v$ and $u$ in the subgraph of $G$ induced by $desc(w)$. We have the following.

\begin{lemma}\label{L:shortcut-chain}
For any $w,v$ with $w\otree v$, there exists a shortcut chain from $w$ to $v$ of length $d_G^{w}(w,v)$.
\end{lemma}

\begin{proof}
Let $p$ be a shortest path from $w$ to $v$ passing only through the descendants of $w$.
We show that $p$ can be decomposed into subpaths $p_i$ between vertices $w=v_1 \ou \ldots \ou v_n=v$ such that each $p_i$ is a valley path.
Once shown the claim follows.

Let $\preceq$ be a total ordering on $V(G)$ extending $\ou$.
By Lemma~\ref{L:ch-partial} it suffices to show that the subpaths $p_i$ are valley paths w.r.t. $\ou$.
Let $u\in V(p)\setminus\{v,w\}$ be the smallest intermediate vertex w.r.t. $\preceq$.
If $v\prec u$, then $p$ is already a valley path and the decomposition is trivial.
Otherwise, we must have $w\prec u\prec v$ and can decompose $p$ into $p_{wu}$ and $p_{uv}$ connecting $w$ to $u$ and $u$ to $v$, respectively.
As $u$ is the minimal intermediate vertex, $p_{wu}$ is a valley path, and $p_{uv}$ passes only through the descendants of $u$.
By induction on the length of $p$, $p_{uv}$ can be decomposed into a chain of valley paths, which combined with $p_{wu}$ forms a decomposition of $p$.
\end{proof}

\begin{theorem}
Algorithms~\ref{algo:ch-update-dec} and~\ref{algo:ch-update-inc} correctly update distance entries in $\lab$ to contain distances w.r.t. the updated graph.
\end{theorem}

\begin{proof}[Proof (Sketch)]
By Definition~\ref{def:shortcut-chain}, distance entries in $\lab$ are the lengths of minimal shortcut chains.
The first for loop in Algorithms~\ref{algo:ch-update-dec} and~\ref{algo:ch-update-inc} collects all (potentially) affected distance entries  associated with shortcut chains that had their last shortcut updated.
The second for loop iteratively extends those shortcut chains.
\end{proof}

Recall that the subgraph distance  $\sgd$ in Definition~\ref{def:subgraph-distance} is based on the subgraph of $\ch$ containing only vertices between $w$ and $v$. Since shortcut chains between $w$ and $v$ are paths within these subgraphs, we immediately obtain the following corollary.

\begin{corollary}\label{cor:sgd}
Let $w\otree v$. Then $d_G^w(w,v)=\sgd(w,v)$.
\end{corollary}

\vspace{0.1cm}
\noindent\textbf{2-Hop Cover Proof.~}~We show that the hierarchical labelling $\lab$ is indeed a 2-hop labelling~\cite{cohen2003reachability}.
Recall that $L_v[w] = \sgd(w, v)$.

\begin{lemma}\label{L:stable-2-hop}
For any two vertices $s,t\in V(G)$, there exists at least one vertex $r$ with $r \ou s$ and $r \ou t$ s.t. $L_s[r] + L_t[r] = d_G(s,t)$.
\end{lemma}
\begin{proof}
Let $p$ be a shortest path between $s$ and $t$ in $G$, and $r$ the vertex in $p$ with the minimal $\tau(r)$.
Then $r\ou v$ for all $v\in V(p)$. Thus, $p$ lies in the subgraph of $G$ induced by the descendants of $r$. By Corollary~\ref{cor:sgd} $L_s[r]$ and $L_t[r]$ store distances within this subgraph.
It follows that $L_s[r] + L_t[r] \leq |p| = d_G(s,t)$.
The direction $L_s[r] + L_t[r] \geq d_G(s,t)$ is obvious. 
\end{proof}

\vspace{0.1cm}
\noindent\textbf{Complexity Analysis.~}~Let $E_\Delta$ denote the number of affected edges in $G$, $S_\Delta$ the number of affected shortcuts in $\ch$, and $L_\Delta$ the number of affected distance entries in $\lab$. 
Further, denote by $d_{\max}$ the maximum degree of vertices in $\ch$, and by $h$ the maximum number of ancestors of a vetex w.r.t. $\otree$.
We assume constant time access to shortcut weights in $\ch$.

\begin{theorem}
Algorithms~\ref{algo:dch-minus} and~\ref{algo:dch-plus} operate in
$O(E_\Delta + S_\Delta\cdot d_{\max})$ and
$O(E_\Delta\cdot d_{\max} + S_\Delta\cdot d_{\max}^2)$, respectively.
\end{theorem}

\begin{proof}
Consider Algorithm~\ref{algo:dch-minus}. The shortcuts that are \emph{potentially} affected are  affected edges, plus shortcuts incident to affected shortcuts.
The number of the latter is bounded by $O(S_\Delta\cdot d_{\max})$.
For Algorithm~\ref{algo:dch-plus},  the \emph{potentially} affected shortcuts are similarly bounded, but processing each shortcut requires $O(d_{\max})$ time to compute the new distance value.
\end{proof}

\begin{theorem}
Algorithms~\ref{algo:ch-update-dec} and~\ref{algo:ch-update-inc} operate in
$O(S_\Delta\cdot h + L_\Delta\cdot d_{\max})$ and
$O(S_\Delta\cdot h\cdot d_{\max} + L_\Delta\cdot d_{\max}^2)$, respectively.
\end{theorem}

\begin{proof}
Consider Algorithm~\ref{algo:ch-update-dec}.
The distance entries that are \emph{potentially} affected in the first for loop correspond to the affected shortcuts, plus distance entries pointing to ancestors of one endpoint of an affected shortcut.
The number of these is bounded by $O(S_\Delta\cdot h)$.
The second for loop investigates potentially affected distance entries that form a triangle with an affected distance entry and a shortcut; their number is thus bounded by $O(L_\Delta\cdot d_{\max})$.
For Algorithm~\ref{algo:ch-update-inc}, the distance entries being processed are similarly bounded, but processing each of them requires $O(d_{\max})$ time to compute the new distance value.
\end{proof}

Under the assumptions that $E_\Delta \le S_\Delta$ and $S_\Delta\cdot h \le L_\Delta\cdot d_{\max}$, which hold in practice, the above complexity bounds can be simplified to
\begin{align*}
\text{-- Algorithms~\ref{algo:dch-minus} and~\ref{algo:dch-plus}} &: O(S_\Delta\cdot d_{\max}) \text{ and } O(S_\Delta\cdot d_{\max}^2);\\
\text{-- Algorithms~\ref{algo:ch-update-dec} and~\ref{algo:ch-update-inc}} &: O(L_\Delta\cdot d_{\max}) \text{ and } O(L_\Delta\cdot d_{\max}^2).
\end{align*}

Note further that the extra $d_{\max}$ factors for the weight increase algorithms are due to weight calculations for shortcuts or distance entries where a shortest path is eliminated, but other paths of the same length remain.
In practice, such cases are rare and observed performance is much closer to the weight decrease case than the analysis suggests.
Tracking supports as for IncH2H could eliminate these factors from the theoretical bounds.

\begin{table}[ht]
\centering
\caption{Summary of datasets - 10 real-world road networks.}\vspace{-0.2cm}
\label{table:datasets}
\begin{tabular}{| l l | r r r|} 
    \hline
    Network & Region & $|V|$ & $|E|$ & Memory \\
    \hline\hline
    NY & New York City & 264,346 & 733,846 & 17 MB \\
    BAY & San Francisco & 321,270  & 800,172 & 18 MB \\
    COL & Colorado & 435,666 & 1,057,066 & 24 MB \\
    FLA & Florida & 1,070,376 & 2,712,798 & 62 MB \\
    CAL & California & 1,890,815 & 4,657,742 & 107 MB \\
    E & Eastern USA & 3,598,623 & 8,778,114 & 201 MB \\
    W & Western USA & 6,262,104 & 15,248,146 & 349 MB \\
    CTR & Central USA & 14,081,816 & 34,292,496 & 785 MB \\
    USA & United States & 23,947,347 & 58,333,344 & 1.30 GB \\ 
    EUR & Western Europe & 18,010,173 & 42,560,279 &  974 MB \\\hline
\end{tabular}
\end{table}
\begin{table*}[ht]
 \centering
 \caption{Comparing update times of our algorithms and the state-of-the-art algorithms in the batch and single update setting.}\vspace{-0.2cm}
 \label{table:update_time}
 \scalebox{0.93}{
 \begin{tabular}{| l || c c c c | c c c c || c c | c c |}  \hline

    \multirow{3}{*}{Network}&\multicolumn{8}{c||}{Batch Update Setting}&\multicolumn{4}{c|}{Single Update Setting}\\\cline{2-13}
    &\multicolumn{4}{c|}{Increase [ms]}&\multicolumn{4}{c||}{Decrease [ms]}&\multicolumn{2}{c|}{Increase [ms]}&\multicolumn{2}{c|}{Decrease [ms]} \\\cline{2-13}
    &\om$^+_p$&IncH2H$^+_p$&\om$^+$&IncH2H$^+$&\om$^-_p$&IncH2H$^-_p$&\om$^-$&IncH2H$^-$&\om$^+$&IncH2H$^+$&\om$^-$&IncH2H$^-$\\\hline\hline
    NY&0.209&0.234&0.790&2.900&0.116&0.187&0.522&2.006&1.190&4.069&0.847&3.021\\
    BAY&0.153&0.178&0.543&2.498&0.103&0.134&0.394&1.769&1.069&3.910&0.827&2.992\\
    COL&0.257&0.318&0.933&4.613&0.179&0.241&0.696&3.306&2.172&9.693&1.708&7.279\\
    FLA&0.311&0.390&1.906&4.981&0.216&0.320&1.368&3.585&2.546&5.703&2.010&4.531\\
    CAL&0.786&1.185&5.079&20.20&0.539&0.855&3.614&13.89&6.951&27.05&5.447&19.99\\
    E&1.913&2.481&12.20&43.57&1.314&1.820&8.197&29.33&13.22&58.97&9.967&42.28\\
    W&2.420&3.841&18.11&68.99&1.757&2.772&12.69&47.76&19.74&83.27&15.01&60.27\\
    CTR&8.721&15.13&58.72&309.7&5.570&10.75&38.48&213.1&59.56&334.3&42.65&233.2\\
    USA&9.321&18.20&73.59&356.3&6.004&13.06&49.29&239.8&71.19&349.7&51.41&240.1\\
    EUR&5.634&8.283&26.83&96.63&3.273&6.969&17.03&66.97&23.24&87.02&15.99&65.11\\\hline
 \end{tabular}}
\end{table*}
\begin{table*}[ht]
 \centering
 \caption{Comparing query times, labelling sizes and construction times of our~\om with the state-of-the-art method~IncH2H.}\vspace{-0.2cm}
 \label{table:performance_time}
 \scalebox{0.94}{
 \begin{tabular}{| l || c c || r r | r r || r r || r r|}  \hline
	\multirow{2}{*}{Network}&\multicolumn{2}{c||}{Query Time [$\mu$s]}&\multicolumn{2}{c|}{Labelling Size}&\multicolumn{2}{c||}{Shortcuts Size}&\multicolumn{2}{c||}{Const. Time [s]}&\multicolumn{2}{c|}{Affected Labels $L_{\Delta}$ (Million)} \\\cline{2-11}
    & \om & \textsc{IncH2H} & \om & IncH2H & \om & IncH2H & \om & IncH2H & \om & IncH2H \\

   \hline\hline
    NY & 0.287 & 0.913 & 130 MB & 826 MB & 15 MB & 42 MB & 2 & 4&8/31 (0.26)&40/99 (0.40) \\
    BAY & 0.299 & 0.841 & 105 MB & 797 MB & 12 MB & 40 MB & 2 & 3&6/24 (0.25)&43/94 (0.46) \\
    COL & 0.349 & 1.018 & 176 MB & 1.35 GB & 15 MB & 51 MB & 4 & 5&13/41 (0.32)&96/166 (0.58) \\
    FLA & 0.396 & 1.019 & 425 MB & 2.38 GB & 40 MB & 129 MB & 10 & 11&15/97 (0.15)&58/283 (0.20) \\
    CAL & 0.490 & 1.333 & 1.03 GB & 8.12 GB & 73 MB & 233 MB & 25 & 30&42/252 (0.17)&297/1,023 (0.29) \\
    E & 0.630 & 1.683 & 2.92 GB & 20.5 GB & 136 MB & 444 MB & 64 & 74&92/736 (0.13)&864/2,627 (0.33) \\
    W & 0.664 & 1.702 & 4.83 GB & 36.0 GB & 231 MB & 758 MB & 107 & 126&119/1,210 (0.10)&3,430/4,595 (0.75) \\
    CTR & 0.812 & 2.483 & 19.7 GB & 177 GB & 558 MB & 1.77 GB & 455 & 858&331/5,092 (0.07)&3,604/23,250 (0.16) \\
    USA & 0.834 & 3.428 & 35.6 GB & 307 GB & 931 MB & 2.97 GB & 710 & 1,081&458/9,206 (0.05)&944/40,220 (0.02) \\ 
    EUR & 1.185 & 3.888 & 36.4 GB & 320 GB & 733 MB & 2.38 GBff & 907 & 1,254&157/9,511 (0.02)&567/42,300 (0.01) \\\hline
 \end{tabular}}
\end{table*}
\section{Experiments}
We conducted experiments to verify the effectiveness of our proposed solution. All the experiments were performed on a Linux server Intel Xeon W-2175 with 2.50GHz CPU, 28 cores, and 512GB of main memory. All the algorithms were implemented in C++20 and compiled using g++ 9.4.0 with the -O3 option. Distance results are exact for all methods considered, and correctness has been verified using Dijkstra. 

\vspace{0.1cm}
\noindent\textbf{Datasets.}
We use 10 undirected real-world road networks. Nine of these road networks are from the US and publicly available at the webpage of the 9th DIMACS Implementation Challenge \cite{demetrescu2009shortest}, while the other one is from Western Europe managed by PTV AG \cite{ptvplanung}. Table \ref{table:datasets} summarises these datasets where the largest dataset is the whole road network in the USA. 

\vspace{0.1cm}
\noindent\textbf{Baselines.}
We compare our algorithms with the state-of-the-art method IncH2H \cite{zhang2022relative} for distance queries on dynamic road networks: We use IncH2H$^+$, IncH2H$^-$ and IncH2H$^+_p$, IncH2H$^-_p$ to denote the sequential and parallel version of IncH2H for edge weight increase and decrease, respectively. We do not consider DCH~\cite{ouyang2020efficient} in our comparison because their query performance is generally orders of magnitude slower than the hub-based labelling methods. Further, we omit the comparison with DTDHL~\cite{zhang2021dynamic} since it has been significantly outperformed by IncH2H. The code for IncH2H was kindly provided by their authors and implemented in C++. We select the balance partition threshold $\beta = 0.2$ and set the number of threads to 28, the number of available cores.

\subsection{Performance Comparison}\label{performance}

\noindent\emph{\underline{Update Time.}~} We randomly sampled 10 batches for each network, where each batch contains 1,000 updates. Then, for each update $(a, b, \omega)$ in a batch, we increase its weight to $2.0 \times \omega$ and maintain the labelling using algorithms for weight increase and decrease (restore) its weight to original (i.e., to $\omega$) and maintain the labelling using algorithms for weight decrease. We process updates in both batch update setting and single update setting (one by one) and report the average update time over 10 batches in Table~\ref{table:update_time}.

Table~\ref{table:update_time} shows that our algorithms \om$^+$ and \om$^-$ significantly outperform IncH2H$^+$ and IncH2H$^-$. In particular, \om$^+$ and \om$^-$ are about 3-4 times faster in update time compared to IncH2H$^+$ and IncH2H$^-$. The reason behind our outstanding performance is two-fold. Firstly, $\hu$ is constructed based on an order induced by $\hq$ which exploits structure of road networks via minimal balanced cuts. This ensures that \om\ contains fewer labels than IncH2H to begin with.
Secondly, unlike IncH2H whose labels contain distances in $G$, \om\ stores distances to ancestors within induced subgraphs. This further reduces the search space when updating distance entries.
Table~\ref{table:performance_time} shows the difference in number of affected labels updated by \om\ and IncH2H, respectively.
E.g. for NY 8 out of 31 million DHL label entries (26\%) had their distance value changed.
We observe that the fraction of affected labels ($L_{\Delta}/|L|$) tends to be smaller for \om~than for IncH2H, which can be attributed to the reduced search space. We do not report results w.r.t. $E_{\Delta}$ and $S_{\Delta}$ as the cost to maintain $G$ and $\hu$ is negligible compared to $L$. The parallel variants of our algorithms \om$^+_p$ and \om$^-_p$ also significantly outperform IncH2H$^+_p$ and IncH2H$^-_p$ on all datasets.

\begin{figure*}[ht!]
\centering
\includegraphics[width=0.95\textwidth]{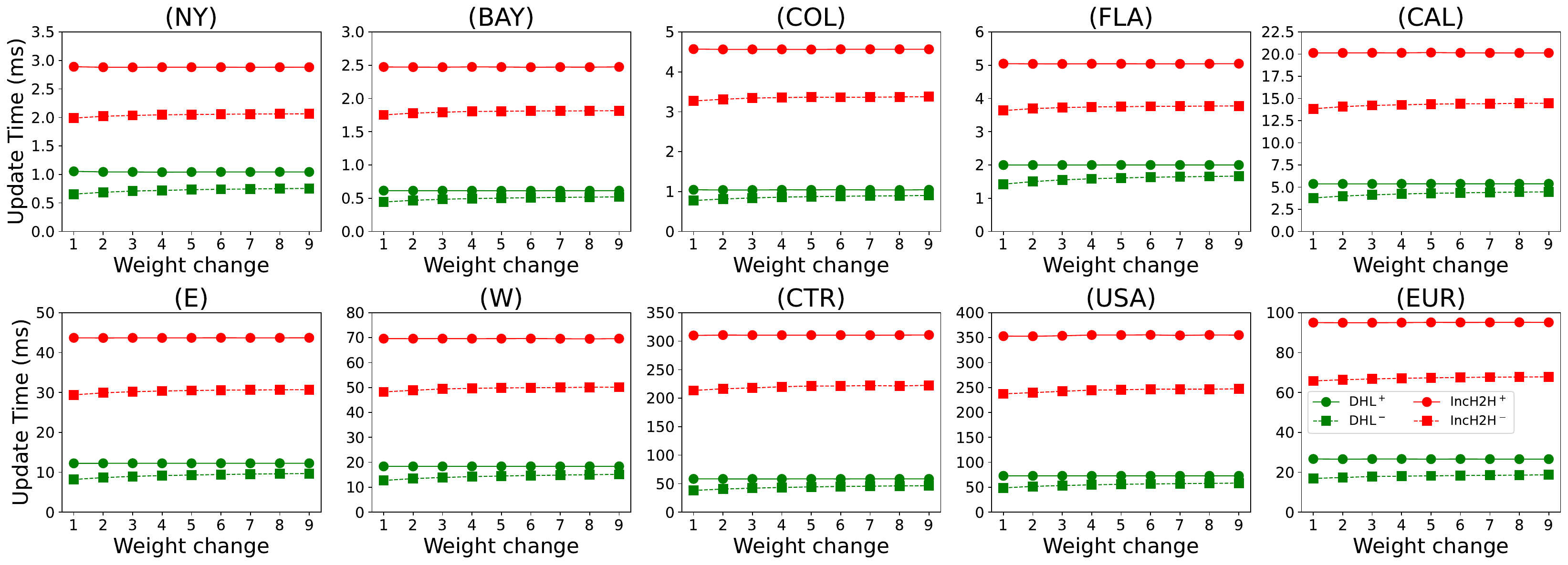}\vspace{-0.3cm}
\caption{Maintenance performance under varying edge weights for both decrease and increase case on all datasets.}
\label{fig:varying_weights}
\includegraphics[width=0.95\textwidth]{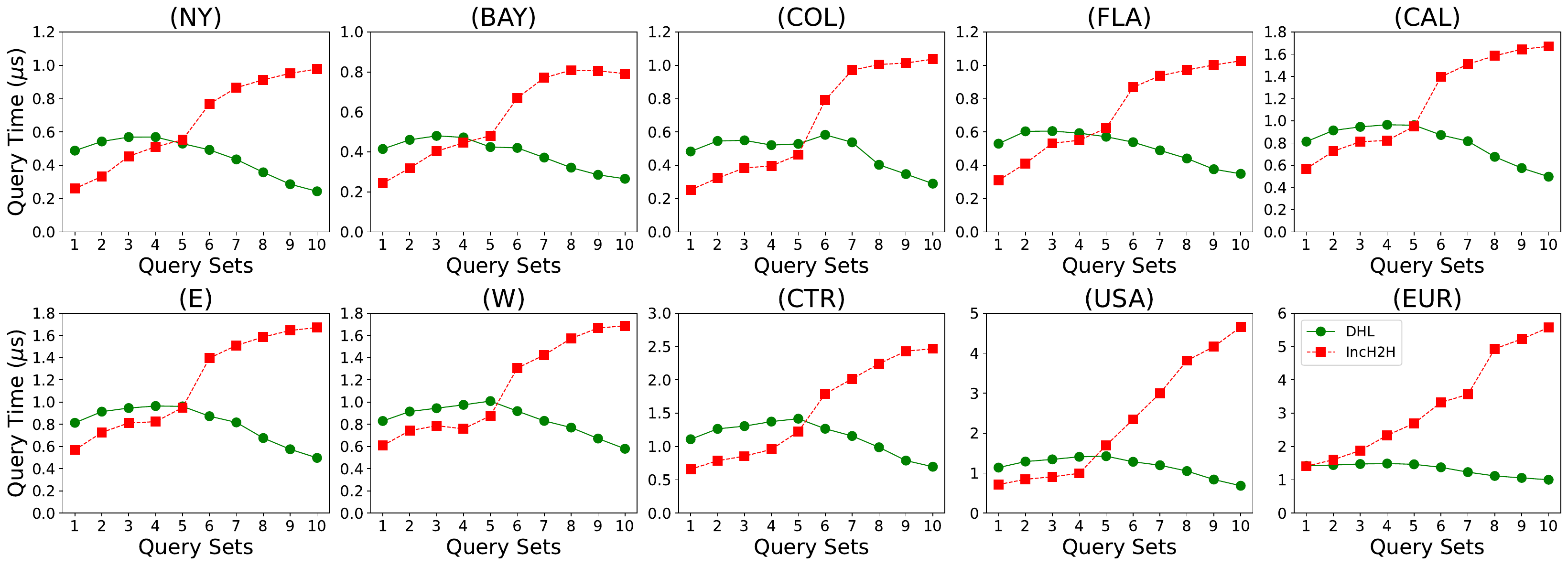}\vspace{-0.3cm}
\caption{Query performance for 10 sets of query pairs with varying distances on all datasets.}
\label{fig:query}\vspace{-0.1cm}
\end{figure*}

\vspace{0.1cm}
\noindent\emph{\underline{Query Time.}~}
We randomly sample 1,000,000 pairs of vertices from each network.  Table~\ref{table:performance_time} reports the average query time over 1 million query pairs. \om~clearly outperforms IncH2H, i.e., about 3 times faster over all road networks. The reason is that \om~ produces significantly smaller labelling size compared to IncH2H. Smaller labels allow \om~to use caching more effectively. Further, \om~processes a significantly smaller number of distance entries related to common ancestors in the labels of a query pair.

We also evaluate the query performance using query pairs with varying distances, similar to \cite{ouyang2018hierarchy,Pohl1969BidirectionalAH}. We generate 10 sets of query pairs $Q_1, Q_2,\dots, Q_{10}$ for each network. Let $x = (\frac{l_{max}}{l_{min}})^{1/10}$, where $l_{min}=\text{1,000}$ and $l_{max}$ is the maximum distance of any query pair. For $\forall1\leq i\leq10$, we generate 10,000 queries in each set $Q_i$, where their distances are in the range $(l_{min} \cdot x^{i-1},\; l_{min} \cdot x^{i}]$. We report the average query time for all road networks in Figure~\ref{fig:query}. We can see that \om~significantly outperforms IncH2H for long distance pairs. Long distance pairs have a significantly small number of common ancestors because their lowest common neighbor generally lies at higher levels of a hierarchy. For short distance pairs, DHL often examines more hops compared to IncH2H because the lowest common ancestor (LCA) lies at lower levels of the query hierarchy, resulting in more ancestors to search in the labels. Despite this, DHL performs comparably on almost all road networks for short distance queries. This is due to two factors: (1) the hops are stored in a continuous memory block, allowing for efficient access, and (2) our data structure enables faster computation of LCA.

\vspace{0.1cm}
\noindent\emph{\underline{Labelling Size.}~}~ We also compare the memory consumed by \om~and the state-of-the-art method IncH2H. Table~\ref{table:performance_time} shows that our method \om~requires significantly less memory than IncH2H. In particular, the labelling size of \om~is about 9 times smaller than IncH2H on the largest three datasets. This is because \om~is constructed based on a vertex partial order from $\hq$. Since $\hq$ is generated using a partitioning technique which produces small and minimal cuts to partition a graph, our distance labels store a smaller number of label entries than IncH2H.
Additionally, IncH2H employs auxiliary data structures to speed up maintenance and ensure strong theoretical bounds, such as support, which inflate memory requirements further.
We also compared the memory sizes of both methods for storing shortcuts, and find that IncH2H uses about 3 times more memory than \om, for similar reasons.

\begin{figure*}[ht]
\centering
\includegraphics[width=0.95\textwidth]{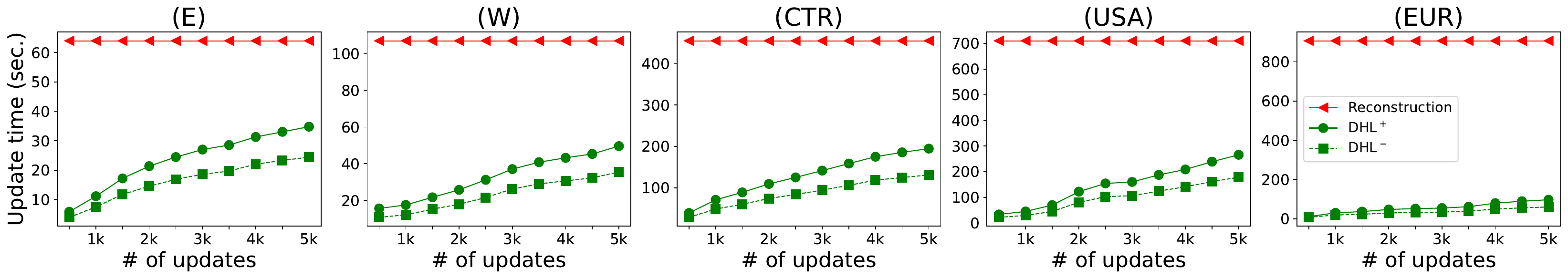}\vspace{-0.3cm}
\caption{Maintenance performance against reconstruction time, evaluated using batches of updates of varying sizes.}\vspace{-0.3cm}
\label{fig:effi}
\end{figure*}
\subsection{Update Time with Varying Weights}
Following~\cite{ouyang2020efficient,zhang2022relative}, we randomly sampled 9 batches, each containing 1,000 edges. We first increase weights of updates $(a, b, \omega)$ in batch $t$ to $(t+1)\times\omega$ and then restore their weights to original i.e., $\omega$ to test the performance of weight increase and decrease case, respectively. Figure~\ref{fig:varying_weights} shows the average update time of our algorithms \om$^+$, \om$^-$ and the state-of-the-art algorithms IncH2H$^+$, IncH2H$^-$.  
We can see that \om$^+$ and \om$^-$ consume significantly less time to update the labelling compared to IncH2H$^+$ and IncH2H$^-$. This is because \om~has significantly reduced labelling sizes compared to IncH2H. As a result, the number of affected labels maintained by our algorithms are much smaller than IncH2H, as shown in Table~\ref{table:performance_time}. 

\subsection{Scalability Test}
We test the scalability of our methods \om$^+$ and \om$^-$ by randomly sampling 5,000 updates for each network and processing them in batches of sizes $\{5, 10$ $ 15, 20, 25, 30, 35, 40, 45, 50\}\times 10^2$. Figure \ref{fig:effi} presents the results for the large 6 datasets E, W, CTR, USA and EUR, while the results for other datasets are omitted as their trends are similar. We process each group with weight increase first and then with weight decrease, which is compared with the time taken by \om~ to construct the labelling from scratch. We see that, even for the largest group with 5,000 updates, the update times of \om$^+$ and \om$^-$ are significantly less than reconstruction.
\section{Extensions}\label{label:extensions}
We discuss how our solution can be extended to directed graphs and to edge/node insertions/deletions, and the boundedness.

\vspace{0.1cm}
\noindent\emph{\underline{Directed Road Networks.}~}
Our method can be easily extended to the directed version of dynamic road networks. We can create \emph{forward} and \emph{reverse} labels for each vertex $v\in V(G)$ to store distances from both directions when constructing the labelling $L$. This may be carried out by running Algorithm~\ref{algo:label-construct} for both forward and reverse directions. Then, we can use \om$^-$ and \om$^+$ twice to maintain $L$, once on the forward labels and again on the backward search. For directed versions of dynamic road networks, existing labelling-based methods require increased memory to store precomputed labels. However, road networks are often nearly undirected, with a few notable exceptions (e.g., Stockholm). In such cases,  two distances stored within each label are often identical. This raises two intriguing possibilities for future research: (1) how to exploit this symmetry to reduce the size of labels in directed road networks, and (2) how to avoid performing separate distance calculations for each direction during both the construction and updating of labels.

\vspace{0.1cm}
\noindent\emph{\underline{Edge/Vertex Insertion/Deletion.}~}
In practice, new roads are rarely created and old roads are rarely deconstructed. Thus, the structure of a road network is considered to be intact. As a result, structural changes such as edge/vertex insertion and deletion are extremely infrequent on road networks. Previous studies address scenarios related to such changes \cite{zhang2022relative,ouyang2018hierarchy,zhang2021dynamic}. Similarly, our algorithms consider such changes in the context of \om~as follows.
Edge deletions can be handled by increasing the weights of the deleted edges to $\infty$ and similarly a vertex deletion can be handled by increasing the weights of its adjacent edges to $\infty$. {For edge insertions, we first update the query hierarchy $\hq$ by identifying the largest affected induced subgraph and recursively repartitioning it to fix the affected nodes in $\hq$. Then, based on the ordering induced by $\hq$, we fix the update hierarchy $\hu$ by invoking the algorithm of \cite{ouyang2020efficient}. Accordingly, we compute new labels of vertices using Algorithm~\ref{algo:label-construct}.

\vspace{0.1cm}
\noindent\emph{\underline{Boundedness.}~}
Compared to IncH2H$^-$ and IncH2H$^+$ which are \emph{relative bounded} and \emph{relative subbounded}~\cite{zhang2022relative}, respectively, \om$^-$ is \emph{relative bounded} and \om$^+$ is  not \emph{relative subbounded} because Algorithm~\ref{algo:ch-update-inc} may recompute label entries if the length of a shortest path increases, even if other shortest paths of the same length remain. IncH2H addresses the challenge by tracking the \emph{support} of labels — a measure conceptually similar to the number of shortest paths — and only updating label entries when the support drops to zero. While this approach reduces unnecessary recomputations, it also increases the index size and adds extra maintenance time. This trade-off would typically be beneficial in cases where there are a large number of unnecessary label recomputations. However, in road networks, shortest paths are often unique, resulting in a low fraction of unnecessary label recomputations. To avoid overhead, we chose not to maintain additional support information. Our experiments confirmed this decision: most distance recomputations were necessary for the networks analyzed.
\vspace{-0.2cm}\section{Conclusion}
We critically examined the limitations of existing state-of-the-art solutions for efficiently answering distance queries on dynamic road networks. To address these limitations, we propose a novel solution called Dual-Hierarchy Labelling (DHL). Our solution integrates two hierarchies of distinct yet complementary data structures designed to enhance query efficiency while minimizing the maintenance time required for label updates.
DHL leverages these hierarchies to balance the trade-off between fast query response times and efficient update processes, which is a persistent challenge in dynamic settings. The experimental results validate the effectiveness of DHL: our approach consistently produces smaller index sizes and achieves faster construction times compared to current state-of-the-art methods. These results highlight the practical value of DHL in real-world applications, demonstrating its potential for scalable and efficient deployment in large-scale, dynamic road networks. This work sets the stage for further exploration into hybrid data structures that can optimize both query speed and maintenance efficiency in increasingly complex network environments.

\vspace{0.2cm}
\begin{acks}
This research was supported partially by the Australian Government through the Australian Research Council's Discovery Projects funding scheme (project DP210102273).
\end{acks}

\bibliographystyle{ACM-Reference-Format}
\balance
\bibliography{references}


\begin{thebibliography}{27}


\ifx \showCODEN    \undefined \def \showCODEN     #1{\unskip}     \fi
\ifx \showDOI      \undefined \def \showDOI       #1{#1}\fi
\ifx \showISBNx    \undefined \def \showISBNx     #1{\unskip}     \fi
\ifx \showISBNxiii \undefined \def \showISBNxiii  #1{\unskip}     \fi
\ifx \showISSN     \undefined \def \showISSN      #1{\unskip}     \fi
\ifx \showLCCN     \undefined \def \showLCCN      #1{\unskip}     \fi
\ifx \shownote     \undefined \def \shownote      #1{#1}          \fi
\ifx \showarticletitle \undefined \def \showarticletitle #1{#1}   \fi
\ifx \showURL      \undefined \def \showURL       {\relax}        \fi
\providecommand\bibfield[2]{#2}
\providecommand\bibinfo[2]{#2}
\providecommand\natexlab[1]{#1}
\providecommand\showeprint[2][]{arXiv:#2}

\bibitem[AG({[n.\,d.]})]%
        {ptvplanung}
\bibfield{author}{\bibinfo{person}{PTV AG}.} \bibinfo{year}{[n.\,d.]}\natexlab{}.
\newblock \bibinfo{title}{Western europe dataset}.
\newblock
\newblock
\urldef\tempurl%
\url{http://www.ptv.de}
\showURL{%
\tempurl}


\bibitem[Bast et~al\mbox{.}(2016)]%
        {bast2016route}
\bibfield{author}{\bibinfo{person}{Hannah Bast}, \bibinfo{person}{Daniel Delling}, \bibinfo{person}{Andrew Goldberg}, \bibinfo{person}{Matthias M{\"u}ller-Hannemann}, \bibinfo{person}{Thomas Pajor}, \bibinfo{person}{Peter Sanders}, \bibinfo{person}{Dorothea Wagner}, {and} \bibinfo{person}{Renato~F Werneck}.} \bibinfo{year}{2016}\natexlab{}.
\newblock \showarticletitle{Route planning in transportation networks}.
\newblock \bibinfo{journal}{\emph{Algorithm engineering: Selected results and surveys}} (\bibinfo{year}{2016}), \bibinfo{pages}{19--80}.
\newblock


\bibitem[Bauer et~al\mbox{.}(2010)]%
        {bauer2010preprocessing}
\bibfield{author}{\bibinfo{person}{Reinhard Bauer}, \bibinfo{person}{Tobias Columbus}, \bibinfo{person}{Bastian Katz}, \bibinfo{person}{Marcus Krug}, {and} \bibinfo{person}{Dorothea Wagner}.} \bibinfo{year}{2010}\natexlab{}.
\newblock \showarticletitle{Preprocessing speed-up techniques is hard}. In \bibinfo{booktitle}{\emph{International Conference on Algorithms and Complexity}}. Springer, \bibinfo{pages}{359--370}.
\newblock


\bibitem[Berry et~al\mbox{.}(2003)]%
        {berry2003minimum}
\bibfield{author}{\bibinfo{person}{Anne Berry}, \bibinfo{person}{Pinar Heggernes}, {and} \bibinfo{person}{Genevieve Simonet}.} \bibinfo{year}{2003}\natexlab{}.
\newblock \showarticletitle{The minimum degree heuristic and the minimal triangulation process}. In \bibinfo{booktitle}{\emph{Graph-Theoretic Concepts in Computer Science: 29th International Workshop, WG 2003. Elspeet, The Netherlands, June 19-21, 2003. Revised Papers 29}}. Springer, \bibinfo{pages}{58--70}.
\newblock


\bibitem[Chen et~al\mbox{.}(2021)]%
        {chen2021p2h}
\bibfield{author}{\bibinfo{person}{Zitong Chen}, \bibinfo{person}{Ada Wai-Chee Fu}, \bibinfo{person}{Minhao Jiang}, \bibinfo{person}{Eric Lo}, {and} \bibinfo{person}{Pengfei Zhang}.} \bibinfo{year}{2021}\natexlab{}.
\newblock \showarticletitle{P2h: Efficient distance querying on road networks by projected vertex separators}. In \bibinfo{booktitle}{\emph{Proceedings of the 2021 International Conference on Management of Data}}. \bibinfo{pages}{313--325}.
\newblock


\bibitem[Cohen et~al\mbox{.}(2003)]%
        {cohen2003reachability}
\bibfield{author}{\bibinfo{person}{Edith Cohen}, \bibinfo{person}{Eran Halperin}, \bibinfo{person}{Haim Kaplan}, {and} \bibinfo{person}{Uri Zwick}.} \bibinfo{year}{2003}\natexlab{}.
\newblock \showarticletitle{Reachability and distance queries via 2-hop labels}.
\newblock \bibinfo{journal}{\emph{SIAM J. Comput.}} \bibinfo{volume}{32}, \bibinfo{number}{5} (\bibinfo{year}{2003}), \bibinfo{pages}{1338--1355}.
\newblock


\bibitem[Demetrescu et~al\mbox{.}(2009)]%
        {demetrescu2009shortest}
\bibfield{author}{\bibinfo{person}{Camil Demetrescu}, \bibinfo{person}{Andrew~V Goldberg}, {and} \bibinfo{person}{David~S Johnson}.} \bibinfo{year}{2009}\natexlab{}.
\newblock \bibinfo{booktitle}{\emph{The shortest path problem: Ninth DIMACS implementation challenge}}. Vol.~\bibinfo{volume}{74}.
\newblock \bibinfo{publisher}{American Mathematical Soc.}
\newblock


\bibitem[Fan and Shi(2010)]%
        {fan2010improvement}
\bibfield{author}{\bibinfo{person}{DongKai Fan} {and} \bibinfo{person}{Ping Shi}.} \bibinfo{year}{2010}\natexlab{}.
\newblock \showarticletitle{Improvement of Dijkstra's algorithm and its application in route planning}. In \bibinfo{booktitle}{\emph{2010 seventh international conference on fuzzy systems and knowledge discovery}}, Vol.~\bibinfo{volume}{4}. IEEE, \bibinfo{pages}{1901--1904}.
\newblock


\bibitem[Farhan et~al\mbox{.}(2023)]%
        {farhan2023hierarchical}
\bibfield{author}{\bibinfo{person}{Muhammad Farhan}, \bibinfo{person}{Henning Koehler}, \bibinfo{person}{Robert Ohms}, {and} \bibinfo{person}{Qing Wang}.} \bibinfo{year}{2023}\natexlab{}.
\newblock \showarticletitle{Hierarchical Cut Labelling--Scaling Up Distance Queries on Road Networks}. In \bibinfo{booktitle}{\emph{Proceedings of the ACM SIGMOD International Conference on Management of Data}}.
\newblock


\bibitem[Geisberger et~al\mbox{.}(2008)]%
        {geisberger2008contraction}
\bibfield{author}{\bibinfo{person}{Robert Geisberger}, \bibinfo{person}{Peter Sanders}, \bibinfo{person}{Dominik Schultes}, {and} \bibinfo{person}{Daniel Delling}.} \bibinfo{year}{2008}\natexlab{}.
\newblock \showarticletitle{Contraction hierarchies: Faster and simpler hierarchical routing in road networks}. In \bibinfo{booktitle}{\emph{International workshop on experimental and efficient algorithms}}. Springer, \bibinfo{pages}{319--333}.
\newblock


\bibitem[Geisberger et~al\mbox{.}(2012)]%
        {geisberger2012exact}
\bibfield{author}{\bibinfo{person}{Robert Geisberger}, \bibinfo{person}{Peter Sanders}, \bibinfo{person}{Dominik Schultes}, {and} \bibinfo{person}{Christian Vetter}.} \bibinfo{year}{2012}\natexlab{}.
\newblock \showarticletitle{Exact routing in large road networks using contraction hierarchies}.
\newblock \bibinfo{journal}{\emph{Transportation Science}} \bibinfo{volume}{46}, \bibinfo{number}{3} (\bibinfo{year}{2012}), \bibinfo{pages}{388--404}.
\newblock


\bibitem[Goldberg and Harrelson(2005)]%
        {goldberg2005computing}
\bibfield{author}{\bibinfo{person}{Andrew~V Goldberg} {and} \bibinfo{person}{Chris Harrelson}.} \bibinfo{year}{2005}\natexlab{}.
\newblock \showarticletitle{Computing the shortest path: A search meets graph theory}. In \bibinfo{booktitle}{\emph{SODA}}, Vol.~\bibinfo{volume}{5}. \bibinfo{pages}{156--165}.
\newblock


\bibitem[Hart et~al\mbox{.}(1968)]%
        {hart1968formal}
\bibfield{author}{\bibinfo{person}{Peter~E Hart}, \bibinfo{person}{Nils~J Nilsson}, {and} \bibinfo{person}{Bertram Raphael}.} \bibinfo{year}{1968}\natexlab{}.
\newblock \showarticletitle{A formal basis for the heuristic determination of minimum cost paths}.
\newblock \bibinfo{journal}{\emph{IEEE transactions on Systems Science and Cybernetics}} \bibinfo{volume}{4}, \bibinfo{number}{2} (\bibinfo{year}{1968}), \bibinfo{pages}{100--107}.
\newblock


\bibitem[Huang et~al\mbox{.}(2021)]%
        {huang2021learning}
\bibfield{author}{\bibinfo{person}{Shuai Huang}, \bibinfo{person}{Yong Wang}, \bibinfo{person}{Tianyu Zhao}, {and} \bibinfo{person}{Guoliang Li}.} \bibinfo{year}{2021}\natexlab{}.
\newblock \showarticletitle{A learning-based method for computing shortest path distances on road networks}. In \bibinfo{booktitle}{\emph{2021 IEEE 37th International Conference on Data Engineering (ICDE)}}. IEEE, \bibinfo{pages}{360--371}.
\newblock


\bibitem[Kriegel et~al\mbox{.}(2007)]%
        {kriegel2007proximity}
\bibfield{author}{\bibinfo{person}{Hans-Peter Kriegel}, \bibinfo{person}{Peer Kr{\"o}ger}, \bibinfo{person}{Peter Kunath}, \bibinfo{person}{Matthias Renz}, {and} \bibinfo{person}{Tim Schmidt}.} \bibinfo{year}{2007}\natexlab{}.
\newblock \showarticletitle{Proximity queries in large traffic networks}. In \bibinfo{booktitle}{\emph{Proceedings of the 15th annual ACM international symposium on Advances in geographic information systems}}. \bibinfo{pages}{1--8}.
\newblock


\bibitem[Ouyang et~al\mbox{.}(2018)]%
        {ouyang2018hierarchy}
\bibfield{author}{\bibinfo{person}{Dian Ouyang}, \bibinfo{person}{Lu Qin}, \bibinfo{person}{Lijun Chang}, \bibinfo{person}{Xuemin Lin}, \bibinfo{person}{Ying Zhang}, {and} \bibinfo{person}{Qing Zhu}.} \bibinfo{year}{2018}\natexlab{}.
\newblock \showarticletitle{When hierarchy meets 2-hop-labeling: Efficient shortest distance queries on road networks}. In \bibinfo{booktitle}{\emph{Proceedings of the 2018 International Conference on Management of Data}}. \bibinfo{pages}{709--724}.
\newblock


\bibitem[Ouyang et~al\mbox{.}(2020)]%
        {ouyang2020efficient}
\bibfield{author}{\bibinfo{person}{Dian Ouyang}, \bibinfo{person}{Long Yuan}, \bibinfo{person}{Lu Qin}, \bibinfo{person}{Lijun Chang}, \bibinfo{person}{Ying Zhang}, {and} \bibinfo{person}{Xuemin Lin}.} \bibinfo{year}{2020}\natexlab{}.
\newblock \showarticletitle{Efficient shortest path index maintenance on dynamic road networks with theoretical guarantees}.
\newblock \bibinfo{journal}{\emph{Proceedings of the VLDB Endowment}} \bibinfo{volume}{13}, \bibinfo{number}{5} (\bibinfo{year}{2020}), \bibinfo{pages}{602--615}.
\newblock


\bibitem[Pohl(1969)]%
        {Pohl1969BidirectionalAH}
\bibfield{author}{\bibinfo{person}{Ira~Sheldon Pohl}.} \bibinfo{year}{1969}\natexlab{}.
\newblock \emph{\bibinfo{title}{Bi-Directional and Heuristic Search in Path Problems}}.
\newblock \bibinfo{thesistype}{Ph.\,D. Dissertation}. \bibinfo{address}{Stanford, CA, USA}.
\newblock
\newblock
\shownote{AAI7001588}.


\bibitem[Tarjan(1983)]%
        {tarjan1983data}
\bibfield{author}{\bibinfo{person}{Robert~Endre Tarjan}.} \bibinfo{year}{1983}\natexlab{}.
\newblock \bibinfo{booktitle}{\emph{Data structures and network algorithms}}.
\newblock \bibinfo{publisher}{SIAM}.
\newblock


\bibitem[Wei et~al\mbox{.}(2020)]%
        {wei2020architecture}
\bibfield{author}{\bibinfo{person}{Victor~Junqiu Wei}, \bibinfo{person}{Raymond Chi-Wing Wong}, {and} \bibinfo{person}{Cheng Long}.} \bibinfo{year}{2020}\natexlab{}.
\newblock \showarticletitle{Architecture-intact oracle for fastest path and time queries on dynamic spatial networks}. In \bibinfo{booktitle}{\emph{Proceedings of the 2020 ACM SIGMOD International Conference on Management of Data}}. \bibinfo{pages}{1841--1856}.
\newblock


\bibitem[Wu et~al\mbox{.}(2012)]%
        {wu2012shortest}
\bibfield{author}{\bibinfo{person}{Lingkun Wu}, \bibinfo{person}{Xiaokui Xiao}, \bibinfo{person}{Dingxiong Deng}, \bibinfo{person}{Gao Cong}, \bibinfo{person}{Andy~Diwen Zhu}, {and} \bibinfo{person}{Shuigeng Zhou}.} \bibinfo{year}{2012}\natexlab{}.
\newblock \showarticletitle{Shortest Path and Distance Queries on Road Networks: An Experimental Evaluation}.
\newblock \bibinfo{journal}{\emph{Proceedings of the VLDB Endowment}} \bibinfo{volume}{5}, \bibinfo{number}{5} (\bibinfo{year}{2012}).
\newblock


\bibitem[Yawalkar and Ranu(2019)]%
        {yawalkar2019route}
\bibfield{author}{\bibinfo{person}{Pranali Yawalkar} {and} \bibinfo{person}{Sayan Ranu}.} \bibinfo{year}{2019}\natexlab{}.
\newblock \showarticletitle{Route recommendations on road networks for arbitrary user preference functions}. In \bibinfo{booktitle}{\emph{2019 IEEE 35th International Conference on Data Engineering (ICDE)}}. IEEE, \bibinfo{pages}{602--613}.
\newblock


\bibitem[Zhang(2021)]%
        {zhang2021efficient}
\bibfield{author}{\bibinfo{person}{Mengxuan Zhang}.} \bibinfo{year}{2021}\natexlab{}.
\newblock \emph{\bibinfo{title}{Efficient shortest path query processing in dynamic road networks}}.
\newblock \bibinfo{thesistype}{Ph.\,D. Dissertation}.
\newblock


\bibitem[Zhang et~al\mbox{.}(2021)]%
        {zhang2021dynamic}
\bibfield{author}{\bibinfo{person}{Mengxuan Zhang}, \bibinfo{person}{Lei Li}, \bibinfo{person}{Wen Hua}, \bibinfo{person}{Rui Mao}, \bibinfo{person}{Pingfu Chao}, {and} \bibinfo{person}{Xiaofang Zhou}.} \bibinfo{year}{2021}\natexlab{}.
\newblock \showarticletitle{Dynamic hub labeling for road networks}. In \bibinfo{booktitle}{\emph{2021 IEEE 37th International Conference on Data Engineering (ICDE)}}. \bibinfo{pages}{336--347}.
\newblock


\bibitem[Zhang and Yu(2022)]%
        {zhang2022relative}
\bibfield{author}{\bibinfo{person}{Yikai Zhang} {and} \bibinfo{person}{Jeffrey~Xu Yu}.} \bibinfo{year}{2022}\natexlab{}.
\newblock \showarticletitle{Relative Subboundedness of Contraction Hierarchy and Hierarchical 2-Hop Index in Dynamic Road Networks}. In \bibinfo{booktitle}{\emph{Proceedings of the 2022 International Conference on Management of Data}}. \bibinfo{pages}{1992--2005}.
\newblock


\bibitem[Zheng et~al\mbox{.}(2010)]%
        {zheng2010understanding}
\bibfield{author}{\bibinfo{person}{Yu Zheng}, \bibinfo{person}{Yukun Chen}, \bibinfo{person}{Quannan Li}, \bibinfo{person}{Xing Xie}, {and} \bibinfo{person}{Wei-Ying Ma}.} \bibinfo{year}{2010}\natexlab{}.
\newblock \showarticletitle{Understanding transportation modes based on GPS data for web applications}.
\newblock \bibinfo{journal}{\emph{ACM Transactions on the Web (TWEB)}} \bibinfo{volume}{4}, \bibinfo{number}{1} (\bibinfo{year}{2010}), \bibinfo{pages}{1--36}.
\newblock


\bibitem[Zhu et~al\mbox{.}(2013)]%
        {zhu2013shortest}
\bibfield{author}{\bibinfo{person}{Andy~Diwen Zhu}, \bibinfo{person}{Hui Ma}, \bibinfo{person}{Xiaokui Xiao}, \bibinfo{person}{Siqiang Luo}, \bibinfo{person}{Youze Tang}, {and} \bibinfo{person}{Shuigeng Zhou}.} \bibinfo{year}{2013}\natexlab{}.
\newblock \showarticletitle{Shortest path and distance queries on road networks: towards bridging theory and practice}. In \bibinfo{booktitle}{\emph{Proceedings of the 2013 ACM SIGMOD International Conference on Management of Data}}. \bibinfo{pages}{857--868}.
\newblock


\end{thebibliography}

\end{document}